\documentclass[a4paper,UKenglish]{lipics-v2019}
\usepackage{bbding}
\usepackage{cite}
\usepackage{comment}
\usepackage{graphicx}
\usepackage{mathpartir}
\usepackage{microtype}
\usepackage{paper}
\usepackage{pifont}
\usepackage{proof}
\usepackage{stmaryrd}
\usepackage{textcomp}
\usepackage{turnstile}
\usepackage{url}
\usepackage{xcolor}
\usepackage{xspace}
\usepackage{xargs}
\usepackage{paralist}
\usepackage{xfrac}

\newcommandx*\seq[5][1=1, 3=n, 4={,}, 5=\dots]{#2_{#1}#4#5#4#2_{#3}}
\newcommandx*\seqx[5][1=1, 3=n, 4={,}, 5=\dots]{{\renewcommand{\i}{#1}#2}#4#5#4{\renewcommand{\i}{#3}#2}}

\newcommand{\alphaprod}{\seq{\TypeA}[n][\times]}
\newcommand{\TypeAprod}{\seq{\TypeA}[n][\times]}

\newcommandx*{\upto}[2][1=1]{#1,\dots,#2}

\definecolor{darkgray}{rgb}{0.31,0.31,0.33}
\definecolor[named]{lipicsGray}{rgb}{0.31,0.31,0.33}
\definecolor[named]{lipicsBulletGray}{rgb}{0.60,0.60,0.61}
\definecolor[named]{lipicsLineGray}{rgb}{0.51,0.50,0.52}
\definecolor[named]{lipicsLightGray}{rgb}{0.85,0.85,0.86}
\definecolor[named]{lipicsYellow}{rgb}{0.99,0.78,0.07}

\definecolor[named]{keyword}{HTML}{e0004a}
\definecolor[named]{constructor}{HTML}{009304}
\definecolor[named]{darkviolet}{HTML}{1b0160}

\lstdefinestyle{splay-style}{%
  basicstyle=\color{darkviolet}\ttfamily\small,
  extendedchars=true,
  alsoletter={*,'},
  keywordstyle=[1]\color{keyword}\small,
  keywordstyle=[2]\color{constructor}\small,
  otherkeywords={(,)},
  keywords=[1]{if,then,else,match,with,let,in},
  keywords=[2]{leaf,true,false,(,)},
  sensitive=true,
  morecomment=[l]{!}, 
  commentstyle=\color{codeGreen}, 
  numberstyle=\tiny\color{codeGray}
}
\lstdefinelanguage{splay}{style=splay-style}
\newcommand{\RED}[1]{{\color{red}#1}}

\title{Type-Based Analysis of Logarithmic Amortised Complexity}
\titlerunning{Logarithmic Amortised Complexity}

\author{Martin Hofmann~$\dag$}{Department of Computer Science\\
  LMU Munich, Germany}{}{}{}

\author{Lorenz Leutgeb}{Institut für Logic and Computation 192/4\\
  Technische Universität Wien}{lorenz@leutgeb.xyz}{https://orcid.org/0000-0003-0391-3430}{}

\author{Georg Moser}{Department of Computer Science \\
  University of Innsbruck, Austria}{georg.moser@uibk.ac.at}{https://orcid.org/0000-0001-9240-6128}{}

\author{David Obwaller}{Department of Computer Science \\
  University of Innsbruck, Austria}{david.obwaller@uibk.ac.at}{https://orcid.org/0000-0002-4925-4744}{}

\author{Florian Zuleger}{Institut für Logic and Computation 192/4\\
  Technische Universität Wien}{zuleger@forsyte.at}{https://orcid.org/0000-0003-1468-8398}{}

\authorrunning{Hofmann et al.}

\Copyright{Lorenz Leutgeb, David Obwaller, Georg Moser, Florian Zuleger}


\EventEditors{John Q. Open and Joan R. Access}
\EventNoEds{2}
\EventLongTitle{42nd Conference on Very Important Topics (CVIT 2016)}
\EventShortTitle{CVIT 2016}
\EventAcronym{CVIT}
\EventYear{2016}
\EventDate{December 24--27, 2016}
\EventLocation{Little Whinging, United Kingdom}
\EventLogo{}
\SeriesVolume{42}
\ArticleNo{23}
\hideLIPIcs  

\nolinenumbers

\begin{document}
\maketitle

\medskip

\begin{center}
  {\Large\textbf{In Memoriam: Martin Hofmann}}
\end{center}

\medskip

\begin{abstract}
  We introduce a novel amortised resource analysis couched in a type-and-effect system.
  Our analysis is formulated in terms of the physicist's method of amortised analysis, and is
  potential-based. The type system makes use of logarithmic potential functions and is the first such system
  to exhibit \emph{logarithmic amortised complexity}.
  With our approach we target the automated analysis of self-adjusting data structures, like splay trees, which so far have only manually been analysed in the literature. In particular, we have implemented a semi-automated prototype, which successfully analyses the zig-zig case of \emph{splaying}, once the type annotations are fixed.
\end{abstract}

\section{Introduction}

Amortised analysis as pioneered by Sleator and Tarjan~\cite{ST:1985,Tarjan:1985} is a method for the worst-case cost analysis of data structures.
The innovation of amortised analysis lies in considering the cost of a single data structure operation as part of a sequence of data structures operations.
The methodology of amortised analysis allows one to assign a low (e.g., constant or logarithmic) amortised cost to a data structure operation even though the worst-case cost of a single operation might be high (e.g., polynomial or worse).
The setup of amortised analysis guarantees that for a sequence of data structure operations the worst-case cost is indeed the number of data structure operations times the amortised cost.
In this way amortised analysis provides a methodology for worst-case cost analysis.

Starting with the initial proposal, one of the objectives of amortised analysis has been to conduct a worst-case cost analysis for self-adjusting binary search trees, such as \emph{splay trees}~\cite{ST:1985,Tarjan:1985}.
These data structures have the behaviour that a single data structure operation might be expensive (ie.\ linear in the size of the tree) but the cost is guaranteed to `average out' in a sequence of data structure operations (ie.\ logarithmic in the size of the tree).
Amortised analysis has been designed to provide a framework for this kind of reasoning on the cost of data structure operations.

The automated cost analysis of imperative, functional, logic and object-oriented programs as well as of more abstract programming paradigms such a term rewriting systems is an active research topic~\cite{AAGP08,ADFG:2010,conf/lpar/BlancHHK10,conf/pldi/GulwaniZ10,AvanziniEM11,AAGP:2011,conf/sas/Alonso-BlasG12,HBCLMMP:2012,HM:2014,ADM15,HM:2015,avanzini2016tct,Flores-Montoya:2017,Giesl:2017,MS:2020}.
Most research has focused on the inference of polynomial bounds on the worst-case cost of the program under analysis.
A few papers also target the inference of exponential and logarithmic bounds~\cite{AAGP08,AvanziniEM11,ChatterjeeFG17,WWC17,conf/fossacs/Kahn020,WM:2020}.
Some of the cited approaches are able to conduct an automated amortised analysis in the sense of Sleator and Tarjan:
The work on type-based cost analysis by Martin Hofmann and his collaborators~\cite{HofmannJ03,JLHSH09,JHLH10,Hoffmann:2011,HAH:2012b,HM:2014,HM:2015,HS:2015a,HS:2015b,JVFH:2017,HDW:2017}, which we discuss in more detail below, directly employs potential functions as proposed by Sleator and Tarjan~\cite{ST:1985,Tarjan:1985}.
For imperative programs, a line of work infers cost bounds from lexicographic ranking functions using arguments that implicitly achieve an amortised analysis~\cite{SZV:2014,SinnZV15,SinnZV17,conf/vmcai/FiedorHRSVZ18} (for details we refer the reader to~\cite{SinnZV17}).
The connection between ranking functions
and amortised analysis has also been discussed in the context of term rewrite systems~\cite{HM:2014}.
Proposals that incorporate amortised analysis within the recurrence relations approach to cost analysis have been discussed in~\cite{conf/sas/Alonso-BlasG12,Flores-Montoya:2017}.
Still, to the best of our knowledge none of the cited approaches is able to conduct a worst-case cost analysis for self-adjusting binary search trees such as splay trees.
One notable exception is~\cite{Nipkow:2015} where
the correct amortised analysis of splay trees~\cite{ST:1985,Tarjan:1985} and other data
structures is certified in \isabelle\
with some tactic support.
However, it is not clear at all if the approach can be further automated.

In this article we take the first step towards the automated analysis of logarithmic amortised cost.
We extend the line of work by Martin Hofmann and his collaborators on amortised analysis, where the search for suitable potential functions is encoded as a type-and-effect systems.
This line of work has lead to several successful tools for deriving accurate bounds on the resource usage of functional \cite{HAH:2012b,JLHSH09,ADM15,AM:2016}, imperative programs \cite{HR:2013,HS14}, as well as term rewriting systems \cite{avanzini2016tct,HM:2014,HM:2015,MS:2020}.
The cited approaches employ a variety of potential functions:
While initially confined to inferring linear cost~\cite{HofmannJ03}, the methods were subsequently extended to cover polynomial \cite{HoffmannH10a}, multivariate polynomial
\cite{HAH:2012b}, and also exponential cost~\cite{HR:2013}.
We for the first time propose a type system that supports logarithmic potential functions, significantly
extending and correcting an earlier note towards this goal~\cite{HM:2018}.

Our analysis is couched in a simple core functional language just sufficiently rich to provide a full definition of our motivating example: \emph{splaying}.
We employ a big-step semantics, following similar approaches in the literature. Further, our type system is geared towards runtime as computation cost (ie.\ we assign a unit cost to each function call and zero cost to every other program statement). It is straightforward to generalise this type system to
other monotone cost models. With respect to non-monotone costs, like eg.\ heap usage, we expect the type system can also be readily be adapted, but this
is outside the scope of the paper.

The type system has been designed with the goal of automation.
As in previous work on type-based amortised analysis, the type system infers constraints on unknown coefficients of template potential functions in a syntax directed way from the program under analysis.
Suitable coefficients can then be found automatically by solving the collected constraints with a suitable constraint solver (ie.\ an SMT solver that supports the theory of linear arithmetic).
The derivation of constraints is straightforward for all syntactic constructs of our programming language.
However, our automated analysis also requires \emph{sharing} and \emph{weakening} rules.
The latter supports the comparison of different potential functions.
As our potential functions are logarithmic, we cannot directly encode the comparison between logarithmic expressions within the theory of linear arithmetic.
Here, we propose several ideas for \emph{linearising} the required comparison of logarithmic expressions. The obtained linear constraints can be then be added to the constraint system.
Our proposed linearisation makes use of
\begin{inparaenum}[(i)]
  \item mathematical facts about the logarithm, referred to as \emph{expert knowledge},
  \item Farkas' lemma for turning the universally-quantified premise of the weakening rule into an existentially-quantified statement that can be added to the constraint system and
  \item finally a subtly modification of Schoenmakers potential.
  \end{inparaenum}

We report on preliminary results for the automated amortised analysis of the splay function.
Our implementation semi-automatically verifies the correctness of a type annotation with logarithmic amortised cost for the splay function
(more specifically for the zig-zig case of the splay function)
using the constraints generated by the type system.
We believe that the ideas presented in this article can be extended beyond the case of splay trees in order to support the analysis of similar self-adjusting data structures such as the ones used in the evaluation of~\cite{Nipkow:2015}.
Summarising, we make the following contributions:
\begin{itemize}
  \item We propose a new class of template potential functions suitable for logarithmic amortised analysis;
      these potential functions in particular include a variant of Schoenmakers' potential (a key building block for the analysis of the splay function)
      and logarithmic expressions.
  \item We present a type-and-effect system for potential-based resource analysis capable of expressing logarithmic amortised costs, and prove its soundness.
  \item We report on a preliminary implementation for the logarithmic amortised analysis of the splay function.
    With respect to the zigzig-case of the splay function, our prototype is able to automatically check
    that the amortised cost is bounded by $3 \log(t) + 1$. All former results to this respect required a manual
    analysis.
\end{itemize}

\paragraph*{Outline}

The rest of this paper is organised as follows.
In Section~\ref{Primer}, we review the key concepts underlying type-based amortised analysis and present our ideas for their extension.
In Section~\ref{Splaying}, we introduce a simple core language underlying our reasoning and provide a full definition of splaying, our running example.
The employed class of potential functions is provided in Section~\ref{ResourceFunctions},
while the type-and-effect system is presented in Section~\ref{Typesystem}. In Section~\ref{Automation} we report on our ideas for implementing the \emph{weakening} rule.
Concretely, we see these ideas at work in Section~\ref{Analysis}, where we employ the established type-and-effect system on the motivating
example of \emph{splaying}. In Section~\ref{sec:implementation}, we present our implementation and design choices in automation.
In Section~\ref{Related}, we present related work and finally, we conclude in Section~\ref{Conclusion}.

\section{Setting the Stage}
\label{Primer}

Our analysis is formulated in terms of the physicist's method of amortised analysis in the style of Sleator and Tarjan~\cite{ST:1985,Tarjan:1985}.
This method assigns a \emph{potential} to data structures of interest and defines the \emph{amortised cost}
of an operation as the sum of the actual cost plus the change of the potential through execution of the operation, ie.\
the central idea of an amortised analysis as formulated by Sleator and Tarjan is to choose a potential function
$\phi$ such that
\begin{equation*}
\phi(v) + a_f(v) = c_f(v) + \phi(f(v))
\tkom
\end{equation*}
holds for all inputs $v$ to a function $f$, where $a_f$, $c_f$ denote the amortised and total cost, respectively,
of executing $f$. Hofmann et al.~\cite{HofmannJ03,HAH11,HAH:2012,HAH:2012b,HM:2014,HM:2015} provide a
generalisation of this idea to a set of potential functions $\phi,\psi$, such that 
\begin{equation*}
\label{eq:AARA}  
\phi(v) \geqslant c_f(v) + \psi(f(v))
\tkom
\end{equation*}
holds for all inputs $v$. This allows to ready off an upper bound on the amortised cost of $f$, ie.\ we have $a_f(v) \leqslant \phi(v)-\psi(v)$. We add that the above inequality indeed generalises the original formulation, which can be seen by setting $\phi(v) \defsym a_f(v) + \psi(v)$.

In this paper, we present a type-based resource analysis based on the idea of potential functions that can infer logarithmic amortised cost.
Following previous work by Hofmann et al.,
we tackle two key problems in order to achieve a semi-automated logarithmic amortised analysis:
\begin{inparaenum}[1)]
  \item Automation is achieved by a \emph{type-and-effect system} that uses \emph{template potential functions}, ie.\ functions of a fixed shape with indeterminate coefficients.
Here, the key challenge is to identify templates that are suitable for logarithmic analysis and that are closed under the basic operations of the considered programming language.
\item In addition to the actual amortised analysis with costs, we employ \emph{cost-free} analysis as a subroutine, setting the amortised $a_f$ and actual costs $c_f$ of all functions $f$ to zero.
This enables a \emph{size analysis} of sorts, because the inequality $\phi(v) \geqslant \psi(f(v))$ bounds the size of the potential $\psi(f(v))$ in terms of the potential $\phi(v)$.
The size analysis we conduct allows lifting the analysis of a subprogram to a larger context, which is crucial for achieving a \emph{compositional analysis}.
\end{inparaenum}
We overview these two aspects in the sequel of the section.

\subsection{Type-and-effect System}

To set the scene, we briefly review amortised analysis formulated as a type-and effect system up to and including the multivariate polynomial analysis, cf.~\cite{JHLH10,HoffmannH10a,HoffmannH10b,HAH11,HAH:2012b,HM:2014,HM:2015,HDW:2017,JVFH:2017}.

\emph{Polynomial Amortised Analysis.}
Suppose that we have types $\TypeA, \TypeB, \TypeC, \dots$ representing sets of values.
We write $\typedomain{\TypeA}$ for the set of values represented by type $\TypeA$.
Types may be constructed from base types such as Booleans and integers, denoted by $\Base$, and by type formers such as list, tree, product, sum, etc.
For each type $\TypeA$, we define a (possibly infinite) set of \emph{basic potential functions} $\BF(\TypeA)\colon \typedomain{\TypeA} \to \Rplus$.
Thus, if $p\in \BF(\TypeA)$ and $v\in \typedomain{\TypeA}$ then $p(v)\in \Rplus$.
An \emph{annotated type} is a pair of a type $\TypeA$ and a function $Q:\BF(\TypeA)\rightarrow \Rplus$ providing a coefficient for each \emph{basic potential function}.
The function $Q$ must be zero on all but finitely many basic potential functions.
For each annotated type $\annotatedtype{\TypeA}{Q}$, the \emph{potential function} $\phi_Q:
\typedomain{\TypeA}\rightarrow \Rplus$ is then given by
\begin{equation*}
  \phi_Q(v) \defsym \sum_{p\in \BF(\TypeA)} Q(p) \cdot p(v) \tpkt
\end{equation*}

By introducing product types, one can regard functions with several arguments as unary functions, which
allows for technically smooth formalisations, cf.~\cite{HoffmannH10a,HoffmannH10b,Hoffmann:2011};
the analyses in the cited papers are called \emph{univariate} as the set of basic potential functions $\BF(\TypeA)$ of a product type $\TypeA$ is given directly.
In the later \emph{multivariate} versions of automated amortised analysis~\cite{HAH11,HAH:2012b,HM:2015} one takes a more fine-grained approach to products.
Namely, one then sets (for arbitrary $n$)
\begin{align*}
\BF(\TypeAprod) & \defsym \seqx{\BF(\TypeA_\i)}[n][\times] \tkom \\
(\seq{p})(\seq{v}) & \defsym \prod_{i = 1}^n p_i(v_i)
\tpkt
\end{align*}
Thus, the basic potential function for a product type is obtained as the multiplication of the basic potential functions of its constituents.%
\footnote{Suppose that for each type $\TypeA$ there exists a distinguished element $u\in \BF(\TypeA)$ with $u(a) = 1$ for all $a\in\typedomain{\TypeA}$. Then, the multivariate product types contain all (linear combinations) of the basic potential functions, extending earlier univariate definitions of product types.}

\emph{Automation.}
The idea behind this setup is that the basic potential functions $\BF(\TypeA)$ are suitably chosen and fixed by the analysis designer, the coefficients $Q(p)$ for $p \in \BF(\TypeA)$, however, are left indeterminate and will (automatically) be fixed during the analysis.
For this, constraints over the unknown coefficients are collected in a syntax-directed way from the function under analysis and then solved by a suitable constraint solver.
The type-and-effect system formalises this collection of constraints as typing rules, where for each construct of the considered programming language a typing rule is given that corresponds to constraints over the coefficients of the annotated types.
Expressing the quest for suitable type annotations as a type-and-effect system allows one to compose typing judgements in a syntax-oriented way without the need for fixing additional intermediate results, which is often required by competing approaches.
This syntax-directed approach to amortised analysis has been demonstrated to work well for datatypes like lists or trees whose basic potential functions are polynomials over the length of a list resp. the number of nodes of a tree.
One of the reasons why this works well is, e.g., that functional programming languages typically include dedicated syntax for list construction and that polynomials are closed under addition by one (ie.\ if $p(n)$ is a polynomial, so is $p(n+1)$), supporting the formulation of a suitable typing rule for list construction, cf.~\cite{HoffmannH10a,HoffmannH10b,Hoffmann:2011,HAH11,HAH:2012b}.
The syntax-directed approach has been shown to generalise from lists and trees to general inductive data types, cf.~\cite{HM:2014,HM:2015,MS:2018a}.

\emph{Logarithmic Amortised Analysis.}
We now motivate the design choices of our type-and-effect system.
The main objective of our approach is the automated analysis of data-structures such as splay trees, which have \emph{logarithmic} amortised cost.
The amortised analysis of splay trees is tricky and requires choosing an adequate potential function:
our work makes use of a variant of Schoenmakers' potential, $\rk(t)$ for a tree $t$, cf.~\cite{Schoenmakers93,Nipkow:2015}, defined inductively by
\begin{align*}
  \rk(\leaf) &\defsym 1,\\
  \rk(\tree{l}{d}{r}) &\defsym \rk(l) + \log(\size{l}) + \log(\size{r}) + \rk(r)
  \tkom
\end{align*}
where $l$, $r$ are the left resp. right child of the tree $\tree{l}{d}{r}$, $\size{t}$ denotes the number of leaves of a tree $t$, and $d$ is some data element that is ignored by the potential function.
Besides Schoenmakers' potential we need to add further basic potential functions to our analysis. This is motivated as follows:
Similar to the polynomial amortised analysis discussed above we want that the basic potential functions can express the construction of a tree, e.g., let us consider the function
$$f(x,d,y) \defsym \tree{x}{d}{y},$$
which constructs the tree $\tree{x}{d}{y}$ from some trees $x,y$ and some data element $d$, and let us assume a constant cost $c_f(x,y) = 1$ for the function $f$.
A type annotation for $f$ is given by
\begin{align*}
&\underbrace{\rk(x) + \log(\size{x}) + \rk(y) + \log(\size{y}) + 1}_{\phi(x,y)}  \geqslant c_f(x,y) + \underbrace{\rk(f(x,d,y))}_{\psi(f(x,y))} \tkom
\end{align*}
ie.\ the potential $\phi(x,y)$ suffices to pay for the cost $c_f$ of executing $f$ and the potential of the result $\psi(f(x,y))$ (the correctness of this type can be established directly from the definition of Schoenmakers' potential).
As mentioned above, the logarithmic expressions in $\phi(x,y)$, ie.\ $\log(\size{x})+\log(\size{y})+1$, specify the \emph{amortised costs} of the operation.

We see that in order to express the potential $\phi(x,y)$ we also need the basic potential functions $\log(\size{t})$ for a tree $t$.
In fact, we will choose the slightly richer set of basic potential functions
\begin{equation*}
p_{(a,b)}(t) = \log(a\size{t}+b) \tkom
\end{equation*}
where $a,b \in \N$ and $t$ is a tree.
We note that by setting $a=0$ and $b=2$ this choice allows us to represent the constant function $u$ with $u(t) = 1$ for all trees $t$.
We further note that this choice of potential functions is sufficiently rich to express that $p_{(a,b)}(t) = p_{(a,b+a)}(s)$ for trees $s,t$ with $\size{t} = \size{s} +1$, which is needed for precisely expressing the change of potential when a tree is extended by one node.
Further, we define basic potential functions for products of trees by setting
\begin{equation*}
p_{(\seq{a},b)}(\seq{t}) = \log(\seqx{a_\i \cdot \size{t_\i}}[n][+]+b) \tkom
\end{equation*}
where $\seq{a},b \in \N$ and $\seq{t}$ is a tuple of trees.
This is sufficiently rich to state the equality
$p_{(a_0,\seq{a},b)}(x_1,\seq{x}) = p_{(a_0+\seq{a},b)}(\seq{x})$,
which supports the formulation of a \emph{sharing} rule, which in turn is needed for supporting the let-construct in functional programming;
cf.~\cite{HAH11,HAH:2012b,HM:2015} for a more detailed exposition on the sharing rule and the let-construct.

\subsection{Cost-free Semantics}

\emph{Polynomial Amortised Analysis.}
We begin by reviewing the cost-free semantics underlying previous work~\cite{HoffmannH10b,Hoffmann:2011,HAH11,HAH:2012b} on polynomial amortised analysis.
Assume that we want to analyse the composed function call
$g(f(\vec{x}),\vec{z})$ using already established analysis results for $f(\vec{x})$ and $g(y,\vec{z})$.
Suppose we have already established that for all $\vec{x}$, $y$, $\vec{z}$ we have:
\begin{align}
&\phi_0(\vec{x}) \geqslant c_f(\vec{x}) + \beta(f(\vec{x}))
\label{AARA:I}
  \\
&\phi_i(\vec{x}) \geqslant \phi'_i(f(\vec{x})) \quad \text{for all $i$ ($0 < i \leqslant n$)}
\label{AARA:II}
  \\
&\beta(y)+\gamma(\vec{z})+\sum_{i=1}^n \phi'_i(y)\phi_i''(\vec{z}) \geqslant c_g(y,\vec{z}) + \psi(g(y,\vec{z}))
\label{AARA:III}
\tkom
\end{align}
where as in the multivariate case above, $n$ is arbitrary and equations~(\ref{AARA:I}) and~(\ref{AARA:III}) assume cost, while equation~(\ref{AARA:II}) is \emph{cost-free}.
Then, we can conclude for all $\vec{x},\vec{z}$ that
\begin{equation*}
  \underbrace{\phi_0(\vec{x})+\gamma(\vec{z})+\sum_{i=1}^n \phi_i(\vec{x})\phi_i''(\vec{z})}_{\phi(\vec{x},\vec{z})} \geqslant c_f(\vec{x})+ c_g(f(\vec{x}),\vec{z}) + \psi(g(f(\vec{x}),\vec{z}))
  \tkom
\end{equation*}
guaranteeing that the potential $\phi(\vec{x},\vec{z})$ suffices to pay for the cost $c_f(\vec{x})$ of computing $f(\vec{x})$, the cost $c_g(f(\vec{x}),\vec{z})$ of computing $g(f(\vec{x}),\vec{z})$ and the potential $\psi(g(f(\vec{x}),\vec{z}))$ of the result $g(f(\vec{x}),\vec{z})$.
We note that the correctness of this inference hinges on the fact that we can multiply equation~\eqref{AARA:II} with $\phi''_i(\vec{z})$ for $i=1\dots n$, using the monotonicity of the multiplication operation (note that potential functions are non-negative).
We highlight that the multiplication argument works well with cost-free semantics, and enables lifting the resource analysis of $f(\vec{x})$ and $g(y,\vec{z})$ to the composed function call $g(f(\vec{x}),\vec{z})$.

\emph{Remark.}
We point out that the above exposition of cost-free semantics in the context of polynomial amortised analysis differs from the motivation given in the literature~\cite{HoffmannH10b,Hoffmann:2011,HAH11,HAH:2012b},
where cost-free semantics are motivated by the quest for \emph{resource polymorphism},
which is the problem of computing (a representation of) all polynomial potential functions (up to a fixed maximal degree) for the program under analysis;
this problem has been deemed of importance for the handling of non tail-recursive programs.
We add that for the amortised cost analysis of inductively generated data-types, the cost-free semantics proved necessary even for handling basic data-structure manipulations~\cite{HM:2014,HM:2015,MS:2020}.
In our view, cost-free semantics incorporate a \emph{size analysis} of sorts.
We observe that equation~\eqref{AARA:II} states that the potential of the result of the evaluation of $f(\vec{x})$ is bounded by the potential of the function arguments $\vec{x}$, without accounting for the costs of this evaluation.
Thus, for suitably chosen potential functions $\phi_i, \phi'_i$ can act as \emph{norms} and capture the size of the result of the evaluation $f(\vec{x})$ in relation to the size of the argument.
As stated above, a separated cost and size analysis enables a compositional analysis, an insight that we also exploit for logarithmic amortised analysis.

\emph{Logarithmic Amortised Analysis.}
Similar to the polynomial case, we want to analyse the composed function call $g(f(\vec{x}),\vec{z})$ using already established analysis results for $f(\vec{x})$ and $g(y,\vec{z})$.
However, now we extend the class of potential functions to sublinear functions.
Assume that we have already established that
\begin{align}
&\phi_0(\vec{x}) \geqslant c_f(\vec{x}) + \beta(f(\vec{x}))
\label{AARA:IV}
\\
&\log(\phi_i(\vec{x})) \geqslant \log(\phi_i'(\vec{x})) \qquad \text{for all $i$ ($0 < i \leqslant n$)}
\label{AARA:V}
\\
&\beta(y)+\gamma(\vec{z})+\sum_{i=1}^n \log( \phi'_i(y)+\phi_i''(\vec{z})) \geqslant c_g(y,\vec{z}) + \psi(g(y,\vec{z}))
\label{AARA:VI}
\tkom
\end{align}
where equations~(\ref{AARA:IV}) and~(\ref{AARA:VI}) assume cost, while equation~(\ref{AARA:V}) is \emph{cost-free}.
Equations~(\ref{AARA:IV}) and~(\ref{AARA:V}) represent the result of an analysis of $f(\vec{x})$ (note that these equations do not contain the parameters $\vec{z}$, which will however be needed for the analysis of $g(f(\vec{x}),\vec{z})$), and equation~(\ref{AARA:VI}) the result of an analysis of $g(y,\vec{z})$.
Then, we can conclude for all $\vec{x},y,\vec{z}$ that
\begin{equation*}
  \underbrace{\phi_0(\vec{x})+\gamma(\vec{z})+\sum_{i=1}^n \log(\phi_i(\vec{x})+\phi_i''(\vec{z}))}_{\phi(\vec{x},\vec{z})} \geqslant c_f(\vec{x}) +  c_g(f(\vec{x}),\vec{z}) + \psi(g(f(\vec{x}),\vec{z}))
  \tkom
\end{equation*}
guaranteeing that the potential $\phi(\vec{x},\vec{z})$ suffices to pay for the cost $c_f(\vec{x})$ of computing $f(\vec{x})$, the cost $c_g(f(\vec{x}),\vec{z})$ of computing $g(f(\vec{x}),\vec{z})$ and the potential $\psi(g(f(\vec{x}),\vec{z}))$ of the result $g(f(\vec{x}),\vec{z})$.
Here, we crucially use monotonicity of the logarithm function, as formalised in Lemma~\ref{lem:log-inequality}.
This reasoning allows us to lift isolated analyses of the functions $f(\vec{x})$ and $g(y,\vec{z})$ to the composed function call $g(f(\vec{x}),\vec{z})$;
this is what is required for a compositional analysis!

\emph{Example.}
We now illustrate the compositional reasoning on an example.
We consider the function $f(x,d,y) \defsym \tree{x}{d}{y}$, which takes two trees $x,y$ and some data element $d$ and returns the tree $\tree{x}{d}{y}$.
Assume that we already have established that
\begin{align}
&\psi(x) + \psi(y) + 1 \geqslant c_f(x,y) + \rk(f(x,d,y))
\label{AARA:VII}
\\
&\log(\size{x}+\size{y}) \geqslant \log(\size{f(x,d,y)})
\label{AARA:VIII} \tkom
\end{align}
where $\psi(u) = \rk(u) + \log(\size{u})$, $c_f(x,y) = 1$, and $d$ is an arbitrary data element, which is not relevant for the cost analysis of $f$.
We now want to analyse the composed function $h(x,a,y,b,z) := f(f(x,a,y),b,z)$.
We will use the above reasoning, instantiating
equations~(\ref{AARA:IV}) and~(\ref{AARA:V}) with equations~(\ref{AARA:VII}) and~(\ref{AARA:VIII}) for the inner function call $f(x,a,y)$, and equation~(\ref{AARA:VI}) with the sum of equations~(\ref{AARA:VII}) and~(\ref{AARA:VIII}) for the outer function call $f(u,b,z)$.
As argued above, we can then conclude for all $x,y,z$ that
\begin{gather*}
  \psi(x) + \psi(y) + \psi(z) + \log(\size{x} + \size{y}) + \log(\size{x} + \size{y} + \size{z}) + 2 \geqslant {}\\
  {} \geqslant c_f(x,a,y) + c_f(f(x,a,y),b,z) + \psi(f(f(x,y),z))
\tkom
\end{gather*}
is a valid resource annotation for $h(x,a,y,b,z) := f(f(x,a,y),b,z)$; we have used equation~(\ref{AARA:VIII}) twice in this derivation, once as $\log(\size{x}+\size{y}) \geqslant \log(\size{f(x,a,y)})$ and once lifted as $\log(\size{x}+\size{y}+\size{z}) \geqslant \log(\size{f(x,a,y)}+\size{z})$.
Kindly note that the above example appears in similar form as part of the analysis of the \lstinline{splay} function described in Section~\ref{Analysis}.

\section{Motivating Example}
\label{Splaying}

In this section, we introduce the syntax of a suitably defined core (first-order) programming language
to be used in the following. Furthermore, we recall the definition of \emph{splaying}, following the presentation
by Nipkow in~\cite{Nipkow:2015}. Splaying constitutes the motivating examples for the type-based logarithmic
amortised resource analysis presented in this paper.

To make the presentation more succinct, we assume only the following types: Booleans $\Bool = \{\true, \false\}$,
an abstract base type $\Base$ (abbrev.\ $\BaseShort$), product types, and binary trees $\Tree$ (abbrev.\ $\TreeShort$),
whose internal nodes are labelled with elements $\typed{b}{\Base}$.
We use lower-case Greek letters for the denotation of types.
Elements $\typed{t}{\Tree}$ are defined by the following grammar which fixes notation.
\begin{equation*}
  t ::= \leaf \mid \tree{t_1}{b}{t_2} \tpkt
\end{equation*}
The size of a tree is the number of leaves: $\size{\leaf} \defsym 1$,
$\size{\tree{t}{a}{u}} \defsym \size{t} + \size{u}$.

Expressions are defined as follows and given in \emph{let normal form} to
simplify the presentation of the semantics and typing rules. In order to ease
the readability, we make use of some mild syntactic sugaring in the presentation
of actual code.

\begin{definition}
\begin{align*}
  \circ & ::= \textup{\lstinline{<}}
	\mid \textup{\lstinline{>}}
	\mid \textup{\lstinline{=}}
  \\
  e & ::= f\ x_1~\dots~x_n \\
    & \mid \true \mid \false \mid e_1 \circ e_2
    && \mid \cif\ x\ \cthen\ e_1\ \celse\ e_2
  \\
    & \mid \textup{\tree{x_1}{x_2}{x_3}} \mid \leaf
    && \mid \match\ x\ \with\
      \textup{\lstinline{|}}\ \leaf\ \arrow e_1\
      \textup{\lstinline{|}}\ \textup{\tree{x_1}{x_2}{x_3}}\ \arrow e_2
  \\
    & \mid \vlet\ x~\equal~e_1\ \vin\ e_2
    && \mid x
\end{align*}
\end{definition}

We skip the standard definition of integer constants $n \in \Z$ as well as variable
declarations, cf.~\cite{Pierce:2002}. Furthermore, we omit binary operations and focus on the bare
essentials for the comparison operators. For the resource analysis these are not of importance,
as long as we assume that no actual costs are emitted.

A \emph{typing context} is a mapping from variables $\VS$ to types.
Type contexts are denoted by upper-case Greek letters.
A program $\Program$ consists of a signature $\FS$ together with a set of function definitions of the form
$f(\seq{x}) = e$, where the $x_i$ are variables and $e$ an expression.
A \emph{substitution} or (\emph{environment}) $\sigma$ is a mapping from variables to values that respects
types. Substitutions are denoted as sets of assignments: $\sigma = \{\seqx{x_\i\mapsto t_\i}\}$.
We write $\dom(\sigma)$ ($\range(\sigma)$) to denote the domain (range) of $\sigma$.
Let $\sigma$, $\tau$ be substitutions such that $\dom(\sigma) \cap \dom(\tau) = \varnothing$. Then we
denote the (disjoint) union of $\sigma$ and $\tau$ as $\sigma \dunion \tau$.
We employ a simple cost-sensitive big-step semantics based on eager evaluation, whose rules are given in Figure~\ref{fig:4}.
The judgement $\eval{\sigma}{\ell}{e}{v}$ means that under environment $\sigma$,
expression $e$ is evaluated to value $v$ in exactly $\ell$ steps. Here only rule applications emit (unit) costs.
If we do not take costs into account, we simply write $\eval{\sigma}{}{e}{v}$.

\begin{figure}[t]
\begin{mathpar}
\infer{\eval{\sigma}{0}{\flstc{false}}{\flstc{false}}}{}  
\and
\infer{\eval{\sigma}{0}{\flstc{true}}{\flstc{true}}}{}
\and
\infer{\eval{\sigma}{0}{\flstc{leaf}}{\flstc{leaf}}}{}
\\
\infer{\eval{\sigma}{0}{\text{\flsttree{$x_1$}{$x_2$}{$x_3$}}}{\text{\flsttree{$t$}{$b$}{$u$}}}}%
{x_1\sigma = t
  &x_2\sigma = b
  &x_3\sigma = u}
\and
\infer{\eval{\sigma}{0}{x}{v}}{x\sigma = v}
\and
\infer
  {\eval{\sigma}{0}{x_1 \circ x_2}{b}}
  {\text{$b$ is value of $x_1\sigma \circ x_2\sigma$}}
\and
\infer
  {\eval{\sigma}{\ell+1}{\flst{f }x_1~\ldots~x_k}{v}}
  {f(x_1,\ldots,x_k) = e \in \Program & \eval{\sigma}{\ell}{e}{v}}
\and
\infer
  {\eval{\sigma}{\ell_1 + \ell_2}{\flstk{let }x\flst{ = }e_1\flstk{ in }e_2}{v}}
  {\eval{\sigma}{\ell_1}{e_1}{w}
  &\eval{\sigma[x \mapsto w]}{\ell_2}{e_2}{v}}
\and
\infer
  {\eval{\sigma}{\ell}{\flstk{match }x\flstk{ with} \begin{array}[t]{l}%
    \flst{| } \flstc{leaf} \flst{ -> } e_1 \\
    \flst{| } \text{\flsttree{$x_0$}{$x_1$}{$x_2$}} \flst{ -> } e_2
  \end{array}}{v}}
  {x\sigma = \flstc{leaf} & \eval{\sigma}{\ell}{e_1}{v}}
\and
\infer
  {\eval{\sigma}{\ell}{\flstk{if }x\flstk{ then }e_1\flstk{ else }e_2}{v}}
  {x\sigma = \flstk{false} & \eval{\sigma}{\ell}{e_2}{v}}
\and
\infer
  {\eval{\sigma}{\ell}{\flstk{match }x\flstk{ with} \begin{array}[t]{l}%
    \flst{| } \flstc{leaf} \flst{ -> }e_1 \\
    \flst{| } \text{\flsttree{$x_0$}{$x_1$}{$x_2$}} \flst{ -> } e_2
  \end{array}}{v}}
  {x\sigma = \text{\flsttree{$t$}{$a$}{$u$}} & \eval{\sigma'}{\ell}{e_2}{v}}
\and
\infer
  {\eval{\sigma}{\ell}{\flstk{if }x\flstk{ then }e_1\flstk{ else }e_2}{v}}
  {x\sigma=\flstk{true} & \eval{\sigma}{\ell}{e_1}{v}}
\end{mathpar}
Here $\sigma[x \mapsto w]$ denotes the update of the environment $\sigma$ such that $\sigma[x \mapsto w](x) = w$ and the value of all other variables remains unchanged. Furthermore, in the second $\match$ rule, we set $\sigma' \defsym \sigma \dunion \{x_0 \mapsto t, x_1 \mapsto a, x_2 \mapsto u\}$.
\caption{Big-Step Semantics}
\label{fig:4}
\end{figure}

\begin{figure}[t]
\centering
\begin{lstlisting}
splay a t = match t with
  | leaf -> leaf
  | (cl, c, cr) ->
   if a = c then (cl, c, cr)
   else if a < c then match cl with
     | leaf -> (cl, c, cr)
     | (bl, b, br) ->
       if a = b then (bl, a, (br, c, cr))
       else if a < b
         then if bl = leaf then (bl, b, (br, c, cr))
           else match splay a bl with
             | (al, a', ar) -> (al, a', (ar, b, (br, c, cr)))
	 else if br = leaf then (bl, b, (br, c, cr))
           else match splay a br with
	     | (al, a', ar)  -> ((bl, b, al), a', (ar, c, cr))
   else match cr with
     | leaf -> (cl, c, cr)
     | (bl, b, br) ->
       if a = b then ((cl, c, bl), a, br)
       else if a < b
         then if bl = leaf then ((cl, c, bl), b, br)
           else match splay a bl with
             | (al, a', ar) -> ((cl, c, al), a', (ar, b, br))
	 else if br = leaf then ((cl, c, bl), b, br)
           else match splay a br with
             | (al, x, xa)  -> (((cl, c, bl), b, al), x, xa)
\end{lstlisting}
\caption{Function \lstinline{splay}.
}
\label{fig:1}
\end{figure}

\emph{Splay trees} have been introduced by Sleator and Tarjan~\cite{ST:1985,Tarjan:1985}
as self-adjusting binary search trees
with strictly increasing in-order traversal. There is no explicit
balancing condition. All operations rely on a tree rotating operation
dubbed \emph{splaying}; \lstinline{splay a t} is performed by rotating element $a$ to the root of tree $t$
while keeping in-order traversal intact. If $a$ is not contained in $t$, then the last element found before
$\leaf$ is rotated to the tree. The complete definition is given in Figure~\ref{fig:1}.
Based on splaying, searching is performed by splaying with the sought element and comparing to the root
of the result. Similarly, the definition of insertion and deletion depends on splaying.
As an example the definition of insertion and delete is given in Figure~\ref{fig:2} and~\ref{fig:3} respectively.
See also~\cite{Nipkow:2015} for full algorithmic, formally verified, descriptions.

All basic operations can be performed in  $\bO(\log n)$ amortised runtime. The logarithmic amortised
complexity is crucially achieved by local rotations of subtrees in the definition of \lstinline{splay}.
Amortised cost analysis of splaying has been provided for example by Sleator and Tarjan~\cite{ST:1985},
Schoenmakers~\cite{Schoenmakers93}, Nipkow~\cite{Nipkow:2015}, Okasaki~\cite{Okasaki:1999}, among
others.
Below, we follow Nipkow's approach, where the actual cost of splaying is measured by counting the
number of calls to $\text{\lstinline{splay}} \colon \BaseShort \times \TreeShort \to \TreeShort$.

\begin{figure}[t]
\begin{lstlisting}
insert a t = if t = leaf then (leaf, a, leaf)
  else match splay a t with
    | (l, a', r) ->
      if a = a' then (l, a, r)
      else if a < a' then (l, a, (leaf, a', r))
	else ((l, a', leaf), a, r)
\end{lstlisting}
\caption{Function \lstinline{insert}.}
\label{fig:2}
\end{figure}

\begin{figure}[t]
\begin{lstlisting}[mathescape]
delete a t = if t = leaf then leaf
  else match splay a t with
    | (l, a', r) ->
      if a = a' then if l = leaf then r
        else match splay_max l with
          | (l', m, r') -> (l', m, r)
	else (l, a', r)

splay_max t = match t with
  | leaf -> leaf
  | (l, b, r) -> match r with
    | leaf -> (l, b, leaf)
    | (rl, c, rr) ->
      if rr = leaf then ((l, b, rl), c, leaf)
        else match splay_max rr with
          | (rrl, x, xa) -> (((l, b, rl), c, rrl), x, xa)
\end{lstlisting}
\caption{Functions \lstinline{delete} and \lstinline{splay_max}.}
\label{fig:3}
\end{figure}

\section{Resource Functions}
\label{ResourceFunctions}

In this section, we detail the basic potential functions employed and clarify the definition
of potentials used.

Only trees are assigned non-zero potential. This is not a severe restriction as potentials for basic datatypes
would only become essential, if the construction of such types would emit actual costs. This is not the case
in our context. Moreover, note that lists can be conceived as trees of particular shape.
The potential $\Phi(t)$ of a tree $t$ is given as a non-negative linear combination of basic functions,
which essentially amount to ``sums of logs'', cf.~Schoenmakers~\cite{Schoenmakers93}.
It suffices to specify the basic functions for the type of trees $\TreeShort$.
As already mentioned in Section~\ref{Primer}, the \emph{rank} $\rk(t)$ of a tree is defined as follows
\begin{align*}
  \rk(\leaf) &\defsym 1\\
  \rk(\tree{t}{a}{u}) &\defsym \rk(t) + \log'(\size{t}) + \log'(\size{u}) + \rk(u)
\tpkt
\end{align*}
We set $\log'(n) \defsym \log_2(\max\{n,1\})$, that is, the (binary) logarithm function is defined for all numbers.
This is merely a technicality, introduced to ease the presentation. In the following, we will denote the modified logarithmic
function, simply as $\log$. Furthermore, recall that $\size{t}$ denotes the number of leaves in tree $t$.
\label{d:log'}
The definition of ``rank'' is inspired by the definition of potential in~\cite{Schoenmakers93,Nipkow:2015},
but subtly changed to suit it to our context.

\goodbreak
\begin{definition}
\label{d:basicpotential}
  The \emph{basic potential functions} of $\Tree$ are
\begin{itemize}
\item $\lambda t. \rk(t)$, or
\item $p_{(a,b)} \defsym \lambda t. \log(a \cdot \size{t} +b )$, where $a,b$ are natural numbers.
\end{itemize}
\end{definition}

The basic functions are denoted as~$\BF$. Note that the
constant function $1$ is representable: $1 = \lambda t. \log(0 \cdot \size{t}+2) = p_{(0, 2)}$.

Following the recipe of the high-level description in Section~\ref{Primer}, potentials
or more generally \emph{resource functions} become definable as linear combinations of basic potential functions.

\begin{definition}
A \emph{resource function} $r \colon \typedomain{\TreeShort} \to \Rplus$ is a non-negative
linear combination of basic potential functions, that is,
\begin{equation*}
  r(t) \defsym \sum_{i \in \N} q_{i} \cdot p_i (t) \tkom
\end{equation*}
where $p_i \in \BF$. The set of resource functions is denoted as $\RF$.
\end{definition}

We employ $\ast$, natural numbers $i$ and pairs of natural numbers $(a,b)_{a,b \in \N}$ as
indices of the employed basic potential functions.
A \emph{resource annotation over $\TreeShort$}, or simply \emph{annotation}, is a
sequence $Q = [q_\ast] \cup [(q_{(a,b)})_{a,b \in \N}]$ with $q_\ast, q_{(a,b)} \in \Qplus$ with all but
finitely many of the coefficients $q_\ast, q_{(a,b)}$ equal to 0.  It represents
a (finite) linear combination of basic potential functions, that is, a resource function.
The empty annotation, that is, the annotation
where all coefficients are set to zero, is denoted as $\varnothing$.

\begin{remark}
We use the convention that the sequence elements of resource annotations
are denoted by the lower-case letter of the annotation,
potentially with corresponding sub- or superscripts.
\end{remark}

\begin{definition}
The \emph{potential} of a tree $t$ with respect to an annotation $Q$,
that is, $Q=[q_\ast] \cup [(q_{(a,b)})_{a,b \in \N}]$, is defined as follows.
\begin{equation*}
  \potential{t}{Q} \defsym q_\ast \cdot \rk(t) + \sum_{a,b \in \N} q_{(a,b)} \cdot p_{(a,b)}(t) \tkom
\end{equation*}
Recall that $p_{(a,b)} = \log(a \cdot \size{t} +b)$ and that $\rk$ is the rank
function, defined above.
\end{definition}

\begin{example}
Let $t$ be a tree, then its potential could be defined as follows:
$\rk(t) + 3 \cdot \log(\size{t}) + 1$. With respect to the above definition this potential
becomes representable by setting $q_\ast \defsym 1, q_{(1,0)} \defsym 3, q_{(0,2)} \defsym 1$.
Thus, $\potential{t}{Q} = \rk(t) + 3 \cdot \log(\size{t}) + 1$.
\qed
\end{example}

We emphasise that the linear combination defined above is not independent. Consider, for
example $\log(2\size{t}+2)=\log(\size{t}+1)+1$.

\emph{Analysis of Products of Trees.}
We now lift the basic potential functions $p_{(a,b)}$ of a single tree to products of trees.
As discussed in Section~\ref{Primer},
we define the potential functions $p_{(\seq{a}[m],b)}$ for a sequence of $m$ trees $\seq{t}[m]$, by setting:
\begin{equation*}
  p_{(\seq{a}[m],b)}(\seq{t}[m]) \defsym \log(\seqx{a_\i \cdot \size{t_\i}}[m][+] + b) \tkom
\end{equation*}
where $\seq{a}[m],b \in \N$.
Equipped with this definition, we generalise annotations to sequences of trees.
An annotation for a sequence of length $m$ is a sequence $Q = [q_1,\dots,q_m] \cup
[(q_{(\seq[m]{a},b)})_{a_i, b \in \N}]$, again vanishing almost everywhere.
Note that an annotation of length $1$ is simply an annotation as defined above, where the coefficient $q_1$ is set
equal to the coefficient $q_\ast$.
Based on this, the potential of a sequence of trees $\seq{t}[m]$ is defined as follows:

\begin{definition}
\label{d:potential}
Let $\seq{t}[m]$ be trees and let
$Q = [\seq{q}[m]] \cup [(q_{(\seq{a}[m],b)})_{a_i, b \in \N}]$
be an annotation of length $m$ as above. We define
\begin{equation*}
\potential{\seq{t}[m]}{Q} \defsym \sum_{i=1}^m q_i \cdot \rk(t_i) + \sum_{\seq{a}[m],b \in \N}
q_{(\seq{a}[m],b)} \cdot p_{(\seq{a}[m],b)}(\seq{t}[m])
  \tkom
\end{equation*}
where $p_{(\seq{a}[m],b)}(\seq{t}[m]) \defsym \log(\seqx{a_\i \cdot \size{t_\i}}[m][+] + b)$ as
defined above.
Note that for an empty sequence of trees, we have $\potential{\epsilon}{Q} \defsym \sum_{b \in \N} q_b \log(b)$.
\end{definition}

Let $t$ be a tree. Note that the rank function $\rk(t)$ amounts to the sum of the logarithms of the
size of subtrees of $t$. In particular if the tree $t$ simplifies to a list of length $n$, then $\rk(t) = (n+1)+\sum_{i=1}^{n} \log(i)$.
Moreover, as $\sum_{i=1}^n \log(i) \in \Theta(n \log n)$, the above defined potential functions are sufficiently rich
to express linear combinations of sub- and super-linear functions.

Let $\sigma$ denote a substitution, let $\Gamma$ denote a typing context
and let $\seqx{\typed{x_\i}{\TreeShort}}[m]$ denote
all tree types in $\Gamma$. A \emph{resource annotation for $\Gamma$} or simply \emph{annotation}
is an annotation for the sequence of trees $\seqx{x_\i\sigma}[m]$.
We define the \emph{potential} of $\potential{\Gamma}{Q}$ with respect to $\sigma$ as
$\spotential{\Gamma}{Q} \defsym \potential{\seqx{x_\i\sigma}[m]}{Q}$.
%

\begin{definition}
An \emph{annotated signature} $\overline{\FS}$
is a mapping from functions $f$ to sets of pairs consisting of the annotation type
for the arguments of $f$, $\annotatedtype{\alphaprod}{Q}$
and the annotation type $\annotatedtype{\beta}{Q'}$ for the result:
\begin{equation*}
  \overline{\FS}(f) \defsym \left\{ \atypdcl{\alphaprod}{Q}{\beta}{Q'} \colon \text{\parbox{50ex}{$f$ takes $n$ arguments of which $m$ are trees, $Q$ is a resource annotation of length $m$ and $Q'$ a resource annotation of length $1$}} \right\} \tpkt
\end{equation*}
Note that $m \leqslant n$ by definition.
\end{definition}

We confuse the signature and the annotated signature and denote
the latter simply as~$\FS$. Instead of $\atypdcl{\alphaprod}{Q}{\beta}{Q'} \in \FS(f)$, we
typically write $\typed{f}{\atypdcl{\alphaprod}{Q}{\beta}{Q'}}$.
As our analysis makes use of a \emph{cost-free semantics} any function symbol is possibly equipped with a \emph{cost-free} signature, independent of~$\FS$. The cost-free signature is denoted as $\FScf$.

\begin{example}
Consider the function \lstinline{splay}: $\BaseShort \times \TreeShort \to \TreeShort$. The induced
annotated signature is given as $\atypdcl{\BaseShort \times \TreeShort}{Q}{\TreeShort}{Q'}$,
where $Q \defsym [q_\ast] \cup [(q_{(a,b)})_{a,b \in \N}]$ and
$Q' \defsym [q'_\ast] \cup [(q'_{(a,b)})_{a,b \in \N}]$.
The logarithmic amortised cost of splaying is then expressible through the following setting:
$q_\ast \defsym 1$, $q_{(1,0)} = 3$, $q_{(0,2)} = 1$,
$q'_\ast \defsym 1$.
All other coefficients are zero.

This amounts to a potential
of the arguments $\rk(t) + 3 \log(\size{t}) + 1$, while for the
result we consider only its rank, that is, the annotation expresses $3 \log(\size{t}) + 1$
as the logarithmic cost of splaying.
The correctness of the induced logarithmic amortised costs for the zig-zig case of splaying
is verified in Section~\ref{Analysis} and is also automatically verified by our prototype.
\qed
\end{example}

Suppose $\potential{\seq{t},u_1,u_2}{Q}$ denotes an annotated sequence of length $n+2$.
Suppose $u_1 = u_2$ and we want to \emph{share} the values $u_i$, that is, the corresponding function arguments
appear multiple times in the body of the function definition. Then we make use of
the operator $\share{Q}$ that adapts the potential suitably. The operator
is also called \emph{sharing operator}.

\begin{lemma}
\label{l:2}
Let $\seq{t},u_1,u_2$ denote a sequence of trees of length $n+2$ with
annotation $Q$. Then there exists a resource annotation $\share{Q}$ such that
$\potential{\seq{t},u_1,u_2}{Q} = \potential{\seq{t},u}{\share{Q}}$, if
$u_1 = u_2 = u$.
\end{lemma}
\begin{proof}
Wlog.\ we assume $n=0$. Thus, let $Q = [q_1,q_2] \cup [(q_{(a_1,a_2,b)})_{a_i \in \N}]$.
By definition
\begin{equation*}
  \potential{u_1,u_2}{Q} = q_1 \cdot \rk(u_1) + q_2 \cdot \rk(u_2) +
  \sum_{a_1,a_2,b \in \N}  q_{(a_1,a_2,b)} \cdot p_{(a_1,a_2,b)}(u_1,u_2)
  \tkom
\end{equation*}
where $p_{(a_1,a_2,b)}(u_1,u_2) = \log(a_1 \cdot \size{u_1} + a_2 \cdot \size{u_2} + b)$.
By assumption $u=u_1=u_2$. Thus, we obtain
\begin{align*}
  \potential{u,u}{Q} & = q_1 \cdot \rk(u) + q_2 \cdot \rk(u) + \sum_{a_1,a_2,b \in \N}  q_{(a_1,a_2,b)} \cdot p_{(a_1,a_2,b)}(u,u)
  \\
  & = (q_1+q_2) \rk(u) + \sum_{a_1+a_2,b \in \N} q_{(a_1+a_2,b)} \cdot p_{(a_1+a_2,b)}(u)
  \\
  &  = \potential{u}{\share{Q}}
    \tkom
\end{align*}
for suitable defined annotation $\share{Q}$, whose definition can be directly read off from the
above constraints.
\end{proof}

We emphasise that the definability of the sharing annotation $\share{Q}$ is based on the fact that
the basic potential functions $p_{(\seq{a}[m],b)}$ have been carefully chosen so that
\begin{equation*}
p_{(a_0,a_1,\seq[2]{a}[m],b)}(x_1,\seq{x}[m]) = p_{(a_0+a_1,\seq[2]{a}[m],b)}(x_1,\seq[2]{x}[m])
  \tkom
\end{equation*}
holds, cf.~Section~\ref{Primer}.

\begin{remark}
  We observe that the proof-theoretic analogue of the sharing operation constitutes in a
  contraction rule, if the type system is conceived as a proof system.
\end{remark}

Let $Q = [q_\ast] \cup [(q_{(a,b)})_{a,b \in \N}]$ be an annotation and let
$K \in \Qplus$. Then we define $Q' \defsym Q + K$ as follows:
$Q' = [q_\ast] \cup [(q'_{(a,b)})_{a,b \in \N}]$,
where $q'_{(0,2)} \defsym q_{(0,2)} + K$ and for all $(a,b) \not= (0,2)$
$q_{(a,b)}' \defsym q_{(a,b)}$. By definition the annotation coefficient $q_{(0,2)}$ is
the coefficient of the basic potential function $p_{(0,2)}(t) = \log(0\size{t} + 2) = 1$, so
the annotation $Q+K$, adds cost $K$ to the potential induced by $Q$.
Further, we define the multiplication of an annotation $Q$ by a constant $K$, denoted as
$K \cdot Q$ pointwise. Moreover, let $P = [p_\ast] \cup [(p_{(a,b)})_{a,b \in \N}]$ be another annotation. Then
the addition $P+Q$ of annotations $P, Q$ is similarly defined pointwise.

\section{Logarithmic Amortised Resource Analysis}
\label{Typesystem}

\begin{figure}[h]
\begin{mathpar}
\infer[\ruleleaf]{
  \tjudge{\varnothing}{Q+k}{\flstc{leaf}}{\TreeShort}{Q'}
}{%
  \forall c \geqslant 2 \ q_{(c)} = \sum_{a+b=c} q'_{(a,b)}
  &
  k = q'_\ast
}
\and
\infer[\ruleshift]{
  \tjudge{\Gamma}{Q+K}{e}{\alpha}{Q'+K}
}{%
  \tjudge{\Gamma}{Q}{e}{\alpha}{Q'}
  &
  K \geqslant 0
}
\and
\infer[\rulewvar]{
  \tjudge{\Gamma, \typed{x}{\alpha}}{Q}{e}{\beta}{Q'}
}{%
  \tjudge{\Gamma}{R}{e}{\beta}{Q'}
  &
  r_i = q_i
  &
  r_{(\veca,b)} = q_{(\veca,0,b)}
}
\and
\infer[\rulenode]{
  \tjudge{\typed{x_1}{\TreeShort},\typed{x_2}{\BaseShort},\typed{x_3}{\TreeShort}}{Q}{\flsttree{x_1}{x_2}{x_3}}{\TreeShort}{Q'}
}{%
  q_1 = q_2 = q'_\ast
  &
  q_{(1,0,0)} = q_{(0,1,0)} = q'_\ast
  &
  q_{(a,a,c)} = q'_{(a,c)}
}
\and
\infer[\rulecmp]{
  \tjudge{\typed{x_1}{\alpha}, \typed{x_2}{\alpha}}{Q}{x_1 \circ x_2}{\BoolShort}{Q}
}{%
  \circ \in \{ <, >, = \}
}
\and
\infer[\ruleite]{
  \tjudge{\Gamma, \typed{x}{\BoolShort}}{Q}{\flstk{if }x\flstk{ then }e_1\flstk{ else }e_2}{\alpha}{Q'}
}{%
  \tjudge{\Gamma}{Q}{e_1}{\alpha}{Q'}
  &
  \tjudge{\Gamma}{Q}{e_2}{\alpha}{Q'}
}
\and
\infer[\rulematch]{
    \tjudge{\Gamma, \typed{x}{\TreeShort}}{Q}{\flstk{match }x\flstk{ with }
    \flst{| }\flstc{leaf}\flst{ -> }e_1
    \flst{| }\flsttree{x_1}{x_2}{x_3}\flst{ -> }e_2}{\alpha}{Q'}
}{%
  \begin{minipage}[b]{25ex}
    $r_{(\veca,a,a,b)} = q_{(\veca,a,b)}$\\[1ex]
    $p_{(\veca,c)} = \sum_{a+b=c} q_{(\veca,a,b)}$\\[2ex]
    $\tjudge{\Gamma}{P+q_{m+1}}{e_1}{\alpha}{Q'}$
  \end{minipage}
  &
  \begin{minipage}[b]{35ex}
      $r_{m+1} = r_{m+2} = q_{m+1}$
    \\[1ex]
    $r_{(\vec{0},1,0,0)} = r_{(\vec{0},0,1,0)} = q_{m+1}$
    \\[2ex]
    $\tjudge{\Gamma, \typed{x_1}{\TreeShort}, \typed{x_2}{\BaseShort}, \typed{x_3}{\TreeShort}}{R}{e_2}{\alpha}{Q'}$
  \end{minipage}
  &
  q_i = r_i = p_i
}
\and
\infer[\rulelettreecf]{
  \tjudge{\Gamma, \Delta}{Q}{\flstk{let }x\flst{ = }e_1\flstk{ in }e_2}{\beta}{Q'}
}{%
  \begin{minipage}[b]{82ex}
    \centering
    $p_i = q_i$ \quad 
    $p_{(\veca,c)} = q_{(\veca, \vec{0}, c)}$ \quad
    $r_j = q_{m+j}$ \quad $r_{k+1} = p'_\ast$ \quad
    $r_{(\vec{0},d,e)} = p'_{(d,e)}$ \quad
    $\forall \vec{b} \ne \vec{0} \left( r_{(\vec{b},0,0)} = q_{(\vec{0},\vec{b},0)} \right)$
    \\[1ex]
    $\forall \vec{b} \ne \vec{0}, \vec{a} \ne \vec{0} \lor c \ne 0\
    \left( q_{(\veca,\vecb,c)} = \sum_{(d,e)} p^{(\vecb,d,e)}_{(\veca,c)} \right)$
    \\[1ex]
    $\forall \vec{b} \ne \vec{0}, d \ne 0 \lor e \neq 0 \
    \left(
    r_{(\vec{b},d,e)} = {p'}^{(\vecb,d,e)}_{(d,e)}
    \wedge
    \forall (d',e') \neq (d,e) \ \left(
    {p'}^{(\vecb,d,e)}_{(d',e')} = 0 \right) \wedge {} \right.$\\
    $\left. {} \land 
    \sum_{(a,c)} p^{(\vecb,d,e)}_{(\veca,c)} \ge
    {p'}^{(\vecb,d,e)}_{(d,e)} \land 
    \forall \veca \ne \vec{0} \lor c \neq 0 \
    \left( p^{(\vecb,d,e)}_{(\veca,c)} \neq 0 \rightarrow  {p'}^{(\vecb,d,e)}_{(d,e)} \leqslant p^{(\vecb,d,e)}_{(\veca,c)} \right) \right)$\\[2ex]
    $\tjudge{\Gamma}{P}{e_1}{\TreeShort}{P'}$
    \hfill
    $\forall {\vecb \ne \vec{0},d \ne 0 \lor e \neq 0} \ \left( \tjudgecf{\Gamma}{P^{(\vecb,d,e)}}{e_1}{\TreeShort}{{P'}^{(\vecb,d,e)}} \right)$
    \hfill
    $\tjudge{\Delta, \typed{x}{\TreeShort}}{R}{e_2}{\beta}{Q'}$
  \end{minipage}
}
\and
\infer[\ruleletgen]{
  \tjudge{\Gamma, \Delta}{Q}{\flstk{let }x\flst{ = }e_1\flstk{ in }e_2}{\beta}{Q'}
}{%
  \begin{minipage}[b]{30ex}
    $p_i = q_i$ \quad $p_{(\veca,c)} = q_{(\veca, \vec{0}, c)}$\\[2ex]
     $\tjudge{\Gamma}{P}{e_1}{\alpha}{\varnothing}$
  \end{minipage}  
  &
  \begin{minipage}[b]{25ex}
    $q_{(\vec{0},\vec{b},c)} = r_{(\vec{b},c)}$ \ ($\vec{b} \not= \vec{0}$) \\[2ex]
    $\tjudge{\Delta, \typed{x}{\alpha}}{R}{e_2}{\beta}{Q'}$
  \end{minipage}
  &
  \begin{minipage}[b]{12ex}
    $r_j = q_{m+j}$\\[1ex]
    $\alpha \not= \TreeShort$
  \end{minipage}
}
\and
\infer[\ruleshare]{
  \tjudge{\Gamma, \typed{z}{\alpha}}{\share{Q}}{e[z,z]}{\beta}{Q'}
}{%
  \tjudge{\Gamma, \typed{x}{\alpha}, \typed{y}{\alpha}}{Q}{e[x,y]}{\beta}{Q'}
}
\and
\infer[\rulew]{%
  \tjudge{\Gamma}{Q}{e}{\alpha}{Q'}
}{%
  \tjudge{\Gamma}{P}{e}{\alpha}{P'}
  &
  \potential{\Gamma}{P} \leqslant \potential{\Gamma}{Q}
  &
  \potential{\Gamma}{P'} \geqslant \potential{\Gamma}{Q'}
}
\and
\infer[\rulevar]{
  \tjudge{\typed{x}{\alpha}}{Q}{x}{\alpha}{Q}
}{%
  \text{$x$ a variable}
}
\and
\infer[\ruleapp]{%
  \tjudge{\typed{x_1}{\alpha_1},\dots,\typed{x_n}{\alpha_n}}{P +  K \cdot Q}{f(\seq{x})}{\beta}{(P' +  K \cdot Q') - 1}
}{%
  \atypdcl{\alpha_1 \times \cdots \times \alpha_n}{P}{\beta}{P'} \in \FS(f)
  &
  \atypdcl{\alpha_1 \times \cdots \times \alpha_n}{Q}{\beta}{Q'} \in \FScf(f)
  &
  K \in \Qplus
}
\end{mathpar}
To ease notation, we set $\veca \defsym \seq[1]{a}[m]$, $\vecb \defsym \seq[1]{b}[k]$ for vectors of indices $a_i, b_j \in \N$. Further,
$i \in \{\upto{m}\}$, $j \in \{\upto{k}\}$ and $a,b,c,d,e \in \N$.
Sequence elements of annotations, which are not constrained are set to zero.
\caption{Type System for Logarithmic Amortised Resource Analysis}
\label{fig:5}
\end{figure}

In this section, we present the central contribution of this work. We delineate
a novel type-and-effect system incorporating a potential-based amortised resource analysis
capable of expressing \emph{logarithmic} amortised costs. Soundness of the approach
is established in Theorem~\ref{t:1}.

Our potential-based amortised resource analysis is couched in a type system,
given in Figure~\ref{fig:5}. If the type judgement
$\tjudge{\Gamma}{Q}{e}{\alpha}{Q'}$ is derivable, then the worst-case cost of evaluating the
expression $e$ is bound from above by the difference between the potential $\spotential{\Gamma}{Q}$
before the execution and the potential $\potential{v}{Q'}$ of the value $v$ obtained through the evaluation of
the expression $e$.
The typing system makes use of a \emph{cost-free}
semantics, which does not attribute any costs to the calculation.
The cost-free typing judgement is denoted as $\tjudgecf{\Gamma}{Q}{e}{\alpha}{Q'}$
and based on a cost-free variant of the application rule, denoted as \ruleappcf.
The rule~\ruleappcf\ is defined as the rule~\ruleapp, however, no costs are accounted for.
Wrt.\ the cost-free semantics, the \emph{empty signature}, denoted as
$\atypdcl{\alphaprod}{\varnothing}{\beta}{\varnothing}$, is always admissible.

\begin{remark}
Note, that if $\atypdcl{\alphaprod}{P}{\beta}{P'}$ and $\atypdcl{\alphaprod}{Q}{\beta}{Q'}$ are
both cost-free signatures for a function $f$, then any linear combination is admissable as
cost-free signature of $f$. Ie.\ we can assume $\atypdcl{\alphaprod}{K \cdot P + L \cdot Q}{\beta}{K \cdot P' + L \cdot Q'}
\in \FScf(f)$, where $K,L \in \Qplus$.  
\end{remark}

\begin{remark}
Principally the type system can be parameterised
in the resource metric (see~e.g.\cite{HAH:2012b}). In this paper, we focus on amortised and worst-case runtime complexity,
symbolically measured through the number of function applications.
 It is straightforward to generalise this type system to other monotone cost models. Wrt.\ non-monotone costs, like eg.\ heap usage, we expect the type system can also be readily be adapted, but this is outside the scope of the paper.
\end{remark}

We consider the typing rules in turn; recall the convention that sequence elements of annotations are denoted by the lower-case letter of the annotation.
Further, note that sequence elements which do not occur in any constraint are set to zero.
The variable rule \rulevar\ types a variable of unspecified
type $\alpha$. As no actual costs are required the annotation is unchanged.
Similarly no resources are lost through the use of control operators. Hence the definition
of the rules \rulecmp\ and \ruleite\ is straightforward.

As exemplary constructor rules, we have rule \ruleleaf\ for the
empty tree and rule \rulenode\ for the node constructor. Both rules
define suitable constraints on the resource annotations to guarantee that the potential of the
values is correctly represented.

The application rule \ruleapp\ represents the application of a rule given in~$\Program$.
Each application emits actual cost $1$, which
is indicated in the addition of $1$ to the annotation $Q$.
In its simplest form, that is, for the factor $K=1$, the rule allows
to directly read off the required annotations for the typing context

In the pattern matching rule \rulematch\ the potential freed through the
destruction of the tree construction is added to the annotation $R$, which
is used in the right premise of the rule. Note that the length of
the annotation $R$ is $m+2$, where $m$ equals the number of tree types in the type context $\Gamma$.

The constraints expressed in the typing rules~\rulelettreecf\ and~\ruleletgen,
guarantee that the potential provided
through annotation $Q$ is distributed among the call to $e_1$ and $e_2$, that is, this rule takes care of function composition.
The numbers $m$, $k$, respectively, denote the number of tree types in $\Gamma$, $\Delta$.
Due to the sharing rule---discussed in a moment---we can assume wlog.\ that each variable in $e_1$ and $e_2$ occurs at most once.

First, consider the rule~\ruleletgen, that is, the expression $e_1$ evaluates to a value $w$ of arbitrary type $\alpha \not= \Tree$. In
this case the resulting value $w$ cannot carry any potential. This is indicated through the empty annotation $\varnothing$ in the typing
judgement $\tjudge{\Gamma}{P}{e_1}{\alpha}{\varnothing}$. Similarly, in the judgement
$\tjudge{\Delta, \typed{x}{\alpha}}{R}{e_2}{\beta}{Q'}$ for the expression $e_2$, all available potential prior to the
execution of $e_2$ stems from the potential embodied in the type context $\Delta$ wrt.\ annotation $Q$. This is enforced by the corresponding
constraints.
Suppose for $\veca \not= \vec{0}$ and $\vecb \not= \vec{0}$, $q_{(\veca,\vecb,c)}$ would be non-zero. Then the corresponding shared potential between
the contexts $\Gamma$ and $\Delta$ wrt.\ $Q$ is discarded by the rule, as there is no possibility this potential is attached to
the result type $\alpha$.

\label{page:letrule}
Second, consider the more involved rule~\rulelettreecf. To explain this rule, we momentarily assume that in $Q$ no potential is shared, that is,
$q_{(\veca,\vecb,c)} = 0$, whenever $\veca\not= \vec{0}, \vecb \not= \vec{0}$. In this sub-case the rule can be simplified as follows:
\begin{equation*}
\infer[\rulelettree]{
  \tjudge{\Gamma, \Delta}{Q}{\flstk{let }x\flst{ = }e_1\flstk{ in }e_2}{\beta}{Q'}
}{%
  \begin{minipage}[b]{20ex}
    $p_i = q_i$\\[1ex]
    $p_{(\veca,c)} = q_{(\veca, \vec{0}, c)}$\\[2ex]
    $\tjudge{\Gamma}{P}{e_1}{\TreeShort}{P'}$
  \end{minipage}
  &
  \begin{minipage}[b]{30ex}
    $r_{(\vec{b},0,c)} = q_{(\vec{0},\vec{b},c)}$ \ ($\vec{b} \not= \vec{0}$) \\[2ex]
    $\tjudge{\Delta, \typed{x}{\TreeShort}}{R}{e_2}{\beta}{Q'}$
  \end{minipage}
  &
  \begin{minipage}[b]{20ex}
    $r_j = q_{m+j}$\\[1ex]
    $r_{k+1} = p'_\ast$\\[1ex]
    $r_{(\vec{0},a,c)} = p'_{(a,c)}$
  \end{minipage}
}
\end{equation*}
Again the potential in $\Gamma, \Delta$ (wrt.\ annotation $Q$) is distributed for the typing of the expressions $e_1$, $e_2$, respectively, which is
governed by the constraints on the annotations.
The simplified rule is obtained, as the assumption that no shared potential exists, makes almost all constraints
vacuous. In particular, the cost-free derivation
$\tjudgecf{\Gamma}{P^{(\vecb,d,e)}}{e_1}{\TreeShort}{{P'}^{(\vecb,d,e)}}$
is not required.

Finally, consider the most involved sub-case, where shared potentials are possible. Contrary to the simplified rules discussed above, such shared potential
cannot be split between the type contexts $\Gamma$ and $\Delta$, respectively. Thus, the full rule necessarily employs the \emph{cost-free semantics}.
Consequently, the premise $\tjudgecf{\Gamma}{P^{(\vecb,d,e)}}{e_1}{\alpha}{{P'}^{(\vecb,d,e)}}$ expresses
that for all non-zero vectors $\vecb$ and arbitrary indices $d$, $e$, the potentials $\potential{\Gamma}{P^{(\vecb,d,e)}}$ suffices
to cover the potential $\potential{\alpha}{{P'}^{(\vecb,d,e)}}$, if no extra costs are emitted (compare~Section~\ref{Primer}). Intuitively this
represents that the values do not increase during the evaluation of $e_1$ to value $w$.

At last, the type system makes use of structural rules, like the \emph{sharing} rule
\ruleshare\ and the weakening rules \rulewvar\ and \rulew. The
\emph{sharing} rule employs the sharing operator, defined in Lemma~\ref{l:2}. Note that
the variables $x,y$ introduced in the assumption of the typing rule are fresh variables, that do not occur
in $\Gamma$. Similarly, the rule \ruleshift\ allows to shift the potential before and after evaluation of the
expression $e$ by a constant $K$.

Note that the weakening rules embody changes in the potential of the type context of expressions considered. This
amounts to the comparison on logarithmic expressions, principally a non-trivial task that cannot be directly
represented as constraints in the type system. Instead, the rule~\rulew\ employs a symbolic potential expressions for these
comparisons, replacing actual values for tree by variables.
Let $\Gamma$ denote a type context containing the type declarations $\typed{x_1}{\TreeShort}, \dots, \typed{x_m}{\TreeShort}$
and let $Q$ be an annotation of length $m$. Then the \emph{symbolic potential}, denoted
as $\potential{\Gamma}{Q}$, is defined as follows.
\begin{equation*}
\potential{\seq{x}[m]}{Q} \defsym \sum_{i=1}^m q_i \cdot \rk(x_i) + \sum_{\seq{a}[m],b \in \N}
q_{(\seq{a}[m],b)} \cdot p_{(\seq{a}[m],b)}(\seq{x}[m])
\tkom
\end{equation*}
where $p_{(\seq{a}[m],b)}(\seq{x}[m]) = \log(\seqx{a_\i \cdot \size{x_\i}}[m][+] + b)$.
In order to actually solve these constraints over symbolic potentials, we have to \emph{linearise} the underlying
comparisons of logarithmic expressions. This is taken up again in Section~\ref{Automation}.

\begin{definition}
\label{d:welltyped}
A program $\Program$ is called \emph{well-typed} if for any rule
$f(\seq{x}[k]) = e \in \Program$ and any annotated signature
$\atypdcl{\seq{\alpha}[k][\times]}{Q}{\beta}{Q'} \in \FS(f)$, we have
$\tjudge{\typed{x_1}{\alpha_1},\dots,\typed{x_k}{\alpha_k}}{Q}{e}{\beta}{Q'}$.
A program $\Program$ is called \emph{cost-free} well-typed, if the
cost-free typing relation is employed.
\end{definition}

Before we state and proof the soundness of the presented type-and-effect system, we establish
the following auxiliary result, employed in the correct assessment of the transfer of potential in the case of
function composition, see Figure~\ref{fig:5}. See also the high-level description provided in Section~\ref{Primer}.

\begin{lemma}
\label{lem:log-inequality}
Assume $\sum_i q_i \log a_i \geqslant q \log b$ for some rational numbers $a_i,b > 0$ and $q_i \geqslant q$.
Then, $\sum_i q_i \log (a_i + c) \geqslant q \log (b + c)$ for all $c \geqslant 1$.
\end{lemma}
\begin{proof}
Wlog.\ we can assume $q=1$ and $q_i \geqslant 1$, as otherwise we simply divide the assumed inequality by $q$.
Further, observe that the assumption $\sum_i q_i \log a_i \geqslant q \log b$ is equivalent to
\begin{equation}
  \label{eq:inequality}
  \prod_i a_i^{q_i} \geqslant b
  \tpkt
\end{equation}
First, we prove that
\begin{equation}
  \label{eq:inequality:ii}
  (x + y)^r \geqslant x^r + y^r \quad r \geqslant 1 \quad x,y \geqslant 0
  \tpkt
\end{equation}
This is proved as follows. Fix some $x \geqslant 0$ and consider $(x + y)^r$ and $x^r + y^r$ as functions in $y$.
It is then sufficient to observe that $(x + y)^r \geqslant x^r + y^r$ for $y = 0$ and that $\sfrac{d}{dy}(x + y)^r \geqslant \sfrac{d}{dy}(x^r + y^r)$ (the derivatives with regard to $y$) for all $y \geqslant 0$.
Indeed, we have $\sfrac{d}{dy}(x + y)^r = r(x + y)^{r-1}$ and $\sfrac{d}{dy}(x^r + y^r) = r y^{r-1}$.
Because of $r \geqslant 1$ and $x \geqslant 0$, we can thus deduce that $\sfrac{d}{dy}(x + y)^r \geqslant \sfrac{d}{dy}(x^r + y^r)$ for all $y \geqslant 0$.

Now we consider some $c \geqslant 1$.
Combining~\eqref{eq:inequality} and~\eqref{eq:inequality:ii}, we get
\begin{equation*}
  \prod_i (a_i+c)^{q_i} \geqslant \prod_i (a_i^{q_i} + c^{q_i}) \geqslant \prod_i a_i^{q_i} +  \prod_i c^{q_i} \geqslant b + c
  \tkom
\end{equation*}
where we have used that $i \ge 1$, and that $q_i \geqslant 1$ and $c \geqslant 1$ imply $\prod_i c^{q_i} \geqslant c$.
By taking the logarithm on both sides of the inequality we obtain the claim.
\end{proof}

Finally, we obtain the following soundness result, which roughly states that if a program $\Program$ terminates, then the difference in potential has paid its execution costs.%
\footnote{A stated, soundness assumes termination of~$\Program$, but our analysis is not restricted to terminating programs. In order to avoid the assumption the soundness theorem would have to be formulated wrt.\ to a partial big-step or a small step semantics, cf.~\cite{HoffmannH10b,MS:2020}. We consider this outside the scope of this work.}

\begin{theorem}[Soundness Theorem]
\label{t:1}
Let $\Program$ be well-typed and let $\sigma$ be a substitution. Suppose $\tjudge{\Gamma}{Q}{e}{\alpha}{Q'}$ and
$\eval{\sigma}{\ell}{e}{v}$. Then $\spotential{\Gamma}{Q} - \potential{v}{Q'} \geqslant \ell$. Further, if $\tjudgecf{\Gamma}{Q}{e}{\alpha}{Q'}$, then $\spotential{\Gamma}{Q} \geqslant \potential{v}{Q'}$.
\end{theorem}
\begin{proof}
The proof embodies the high-level description
given in Section~\ref{Primer}. It proceeds by main induction on $\Pi \colon \eval{\sigma}{\ell}{e}{v}$
and by side induction on $\Xi \colon \tjudge{\Gamma}{Q}{e}{\alpha}{Q'}$, where the latter is employed
in the context of the weakening rules.
We consider only a few cases of interest.
For example, for a case not covered: the variable rule \rulevar\ types a variable of unspecified
type $\alpha$. As no actual costs are required the annotation is unchanged and the
theorem follows trivially.

\emph{Case}. $\Pi$ derives $\eval{\sigma}{0}{\leaf}{\leaf}$. Then $\Xi$
consists of a single application of the rule \ruleleaf:
\begin{equation*}
\infer[\ruleleaf]{\tjudge{\varnothing}{Q+K}{\leaf}{\TreeShort}{Q'}}{%
  \forall c \geqslant 2 \ q_{(c)} = \sum_{a+b=c} q'_{(a,b)}
  &
  K = q'_\ast
  }
\tpkt
\end{equation*}
By assumption $Q= [(q_{(c)})_{c \in \N}]$ is an annotation for the empty sequence of trees.
On the other hand $Q' = [(q'_{(a,b)})_{a,b \in \N}]$ is an annotation of length $1$. Note that
$\rk(\leaf) = 1$ by definition.
Thus, we obtain:
\begin{align*}
\potential{\epsilon}{Q+K} & = K+ \sum_{c} q_{(c)} \cdot \log(c)\\
& =  K+ \sum_{c \ge 2} q_{(c)} \cdot \log(c)\\
& =  q'_\ast + \sum_{a+b\ge 2} q'_{(a,b)} \cdot \log(a+b)\\
& =  q'_\ast + \sum_{a,b} q'_{(a,b)} \cdot \log(a+b)\\
& = q'_\ast \rk(\leaf) + \sum_{a,b} q'_{(a,b)} p_{(a,b)}(\leaf) = \potential{\leaf}{Q'}
\tpkt
\end{align*}

\emph{Case}. Suppose $\Pi$ has the following from:
\begin{equation*}
  \infer{\eval{\sigma}{0}{\text{\flsttree{$x_1$}{$x_2$}{$x_3$}}}{\text{\flsttree{$t$}{$b$}{$u$}}}}%
  {x_1\sigma = t
  &x_2\sigma = b
  &x_3\sigma = u}
  \tpkt
\end{equation*}
Wlog.\ $\Xi$ consists of a single application of the rule \rulenode:
\begin{equation*}
\infer[\rulenode]{\tjudge{\typed{x_1}{\TreeShort},\typed{x_2}{\BaseShort},\typed{x_3}{\TreeShort}}{Q}{\flsttree{x_1}{x_2}{x_3}}{\TreeShort}{Q'}}{%
    q_1 = q_2 = q'_\ast
    &
    q_{(1,0,0)} = q_{(0,1,0)} = q'_\ast
    &
    q_{(a,a,b)} = q'_{(a,b)}
  }
\end{equation*}
By definition, we have $Q = [q_1, q_2] \cup [(q_{(a_1,a_2,b)})_{a_i,b \in \N}]$ and
$Q' = [q'_\ast] \cup [(q'_{(a',b')})_{a',b' \in \N}]$.
We set $\Gamma \defsym \typed{x_1}{\TreeShort}, \typed{x_2}{\BaseShort}, \typed{x_3}{\TreeShort}$ as well as
$x_1\sigma = u$, $x_2\sigma = b$, and $x_3\sigma = v$.
Thus $\spotential{\Gamma}{Q} = \potential{u,v}{Q}$ and we obtain:
\begin{align*}
\potential{u,v}{Q} & = q_1 \cdot \rk(u) + q_2 \cdot \rk(v) +
                    \sum_{a_1,a_2,b} q_{(a_1,a_2,b)} \cdot \log(a_1 \cdot \size{u} + a_2 \cdot \size{v} + b)
\\
& \geqslant q'_\ast \cdot \rk(u) + q'_\ast \cdot \rk(v) +
                    q_{(1,0,0)} \cdot \log(\size{u}) + q_{(0,1,0)} \cdot \log(\size{v}) + {}
\\
& \quad  {} + \sum_{a,b} q_{(a,a,b)} \cdot \log(a \cdot \size{u} + a \cdot \size{v} + b)
\\
& = q'_\ast \cdot (\rk(u)+ \rk(v)+ \log(\size{u})+\log(\size{v})) + {}\\
& \quad {} + \sum_{a, b} q'_{(a,b)} \cdot \log(a \cdot (\size{u}+\size{v}) +b)
\\
& = q'_\ast \cdot \rk(\tree{u}{b}{v}) + \sum_{a, b} q'_{(a,b)} \cdot p_{(a,b)}(\tree{u}{b}{v})
= \potential{\tree{u}{b}{v}}{Q'}
\tpkt
\end{align*}

\emph{Case}. Suppose $\eval{\sigma}{\ell}{e}{v}$ and let the last rule in $\Xi$ be of the following form:
\begin{equation*}
  \infer{
  \tjudge{\Gamma}{Q+K}{e}{\alpha}{Q'+K}
}{%
  \tjudge{\Gamma}{Q}{e}{\alpha}{Q'}
} \tkom
\end{equation*}
where $K \geqslant 0$. By SIH, we have that $\spotential{\Gamma}{Q} - \potential{v}{Q'} \geqslant \ell$, from
which we immediately obtain:
\begin{equation*}
  \spotential{\Gamma}{Q}+K - \potential{v}{Q'}-K = \spotential{\Gamma}{Q} - \potential{v}{Q'} \geqslant \ell \tpkt
\end{equation*}

\emph{Case}.
Consider the first \rulematch\ rule, where $\Pi$ ends as follows:
\begin{equation*}
  \infer{\eval{\sigma}{\ell}{\flstk{match }x\flstk{ with}
    \flst{| } \flstc{leaf} \flst{ -> } e_1 \\
    \flst{| } \text{\flsttree{$x_0$}{$x_1$}{$x_2$}} \flst{ -> } e_2
  }{v}}%
  {x\sigma = \leaf & \eval{\sigma}{\ell}{e_1}{v}}
  \tpkt
\end{equation*}
Wlog.\ we may assume that $\Xi$ ends with the related application
of the \rulematch\ rule:
\begin{equation*}
  \infer{
    \tjudge{\Gamma, \typed{x}{\TreeShort}}{Q}{\flstk{match }x\flstk{ with }
    \flst{| }\flstc{leaf}\flst{ -> }e_1
    \flst{| }\flsttree{x_1}{x_2}{x_3}\flst{ -> }e_2}{\alpha}{Q'}
}{%
  \begin{minipage}[b]{25ex}
    $r_{(\veca,a,a,b)} = q_{(\veca,a,b)}$\\[1ex]
    $p_{(\veca,c)} = \sum_{a+b=c} q_{(\veca,a,b)}$\\[2ex]
    $\tjudge{\Gamma}{P+q_{m+1}}{e_1}{\alpha}{Q'}$
  \end{minipage}
  &
  \begin{minipage}[b]{35ex}
      $r_{m+1} = r_{m+2} = q_{m+1}$
    \\[1ex]
    $r_{(\vec{0},1,0,0)} = r_{(\vec{0},0,1,0)} = q_{m+1}$
    \\[2ex]
    $\tjudge{\Gamma, \typed{x_1}{\TreeShort}, \typed{x_2}{\BaseShort}, \typed{x_3}{\TreeShort}}{R}{e_2}{\alpha}{Q'}$
  \end{minipage}
  &
  q_i = r_i = p_i
}
\tpkt
\end{equation*}
Let $Q$ be an annotation of length $m+1$ while $Q'$ is of length $1$. Thus annotations $P$, $R$ have
lengths $m$, $m+2$, respectively. We write $\vect \defsym t_1, \dots, t_n$ for
the substitution instances of the variables in $\Gamma$. Further $x\sigma = \leaf$,
where the latter equality follows from the assumption on $\Pi$.
By definition and the constraints given in the rule, we obtain:
\begin{align*}
\spotential{\Gamma,\typed{x}{\TreeShort}}{Q} &= \sum_{i} q_i \rk(t_i) + q_{m+1} \rk(\leaf) +
\sum_{\veca, a, c} q_{(\veca, a, c)} \log(\veca\size{\vect} + a\size{\leaf} + c)
\\
&= \sum_{i} q_i \rk(t_i) + q_{m+1}(\rk(\leaf)) + \sum_{\veca, a, c} q_{(\veca, a, c)} \log(\veca \size{\vect} + a + c)
\\
&= \spotential{\Gamma}{P}+q_{m+1}
\tpkt
\end{align*}
Thus $\spotential{\Gamma,\typed{x}{\TreeShort}}{Q} = \spotential{\Gamma}{P+q_{m+1}}$
and the theorem follows by an application of MIH.

Now, consider the second \rulematch\ rule, that is, $\Pi$ ends as follows:
\begin{equation*}
  \infer
  {\eval{\sigma}{\ell}{\flstk{match }x\flstk{ with}
    \flst{| } \flstc{leaf} \flst{ -> }e_1 \\
    \flst{| } \text{\flsttree{$x_0$}{$x_1$}{$x_2$}} \flst{ -> } e_2
  }{v}}%
  {x\sigma = \text{\flsttree{$t$}{$a$}{$u$}} & \eval{\sigma'}{\ell}{e_2}{v}}
  \tpkt
\end{equation*}
As above, we may assume that $\Xi$ ends with the related application
of the \rulematch\ rule. In this subcase, the assumption on $\Pi$ yields
$t \defsym x\sigma = \tree{u}{b}{v}$.
By definition and the constraints given in the rule, we obtain:
\begin{align*}
\spotential{\Gamma,\typed{x}{\TreeShort}}{Q} &= \sum_{i} q_i \rk(t_i) + q_{m+1} \rk(\tree{u}{b}{v}) +
\sum_{\veca, a, c} q_{(\veca, a, c)} \log(\veca\size{\vect} + a\size{\tree{u}{b}{v}} + c)
\\
&=
\sum_{i} q_i \rk(t_i) + q_{m+1}(\rk(u)+\log(\size{u})+\log(\size{v})+\rk(v)) + {}
\\
&\phantom{=}
{} + \sum_{\veca, a, c} q_{(\veca, a, c)} \log(\veca \size{\vect} + a(\size{u}+\size{v}) + c)
\\
&= \spotential{\Gamma, \typed{x_1}{\TreeShort}, \typed{x_2}{\BaseShort}, \typed{x_3}{\TreeShort}}{R}
\tkom
\end{align*}
where we write $\veca\size{\vect}$ as shorthand to
denote componentwise multiplication.

Thus $\spotential{\Gamma,\typed{x}{\TreeShort}}{Q} = \spotential{\Gamma, \typed{x_1}{\TreeShort}, \typed{x_2}{\BaseShort}, \typed{x_3}{\TreeShort}}{R}$
and the theorem follows by an application of MIH.

\emph{Case}.
Consider the \rulelet\ rule, that is, $\Pi$ ends in the following rule:
\begin{equation*}
  \infer
  {\eval{\sigma}{\ell_1 + \ell_2}{\flstk{let }x\flst{ = }e_1\flstk{ in }e_2}{v}}
  {\eval{\sigma}{\ell_1}{e_1}{w}
  &\eval{\sigma[x \mapsto w]}{\ell_2}{e_2}{v}}
\tkom
\end{equation*}
where $\ell = \ell_1+\ell_2$.
First, we consider the sub-case, where the type of $e_1$ is an arbitrary type $\alpha$ but not of type $\Tree$.
Ie.\ we assume that $\Xi$ ends in the following application of the \ruleletgen-rule
\begin{equation*}
  \infer[\ruleletgen]{
  \tjudge{\Gamma, \Delta}{Q}{\flstk{let }x\flst{ = }e_1\flstk{ in }e_2}{\beta}{Q'}
}{%
  \begin{minipage}[b]{30ex}
    $p_i = q_i$ \quad $p_{(\veca,c)} = q_{(\veca, \vec{0}, c)}$\\[2ex]
     $\tjudge{\Gamma}{P}{e_1}{\alpha}{\varnothing}$
  \end{minipage}  
  &
  \begin{minipage}[b]{25ex}
    $q_{(\vec{0},\vec{b},c)} = r_{(\vec{b},c)}$ \ ($\vec{b} \not= \vec{0}$) \\[2ex]
    $\tjudge{\Delta, \typed{x}{\alpha}}{R}{e_2}{\beta}{Q'}$
  \end{minipage}
  &
  \begin{minipage}[b]{10ex}
    $r_j = q_{m+j}$\\[1ex]
    $\alpha \not= \TreeShort$
  \end{minipage}
}
\tpkt
\end{equation*}
Recall that $\veca = a_1,\dots,a_n$, $\vecb = b_1,\dots,b_m$,
$i \in \{1,\dots,m\}$, $j\in\{1,\dots,k\}$ and
$a_i, b_j, a,b,c,d,e$ are natural numbers. Further, the annotations $Q$, $P$, $R$ are of
length $m+k$, $m$ and $k$, respectively, while the corresponding
resulting annotations $Q'$, $P'$ and $R'$, are of length $1$.

By definition and due to the constraints expressed in the typing rule, we have:
\begin{align*}
  \spotential{\Gamma,\Delta}{Q} & =
  \sum_i q_i \rk(t_i) + \sum_j q_{m+j} \rk(u_j) +
  \sum_{\veca,\vecb,c} q_{(\veca,\vecb,c)} \log(\veca\size{\vect} + \vecb\size{\vecu} + c)
  \\
  \spotential{\Gamma}{P} & = \sum_i q_i \rk(t_i) + \sum_{\veca,c} q_{(\veca, \vec{0}, c)}
\log (\veca \size{\vect} + c)
  \\
  \potential{w}{\varnothing} &= 0
  \\
  \spotential{\Delta,\typed{x}{\Tree}}{R} & = \sum_{j} q_{m+j} \rk(u_j) + r_{k+1} \rk(w) +
\sum_{\vecb,a,c} q_{(\vec{0},\vecb,c)} \log(\vecb\size{\vecu}+c)
  \tkom
\end{align*}
where we set $\vect \defsym \seq[1]{t}[m]$ and $\vecu \defsym \seq[1]{u}[k]$, denoting
the substitution instances of the variables in $\Gamma$, $\Delta$, respectively.
Thus, we obtain
\begin{equation*}
  \spotential{\Gamma,\Delta}{Q} \geqslant \spotential{\Gamma}{P} + \spotential{\Delta,\typed{x}{\alpha}}{R}
  \tpkt
\end{equation*}
By main induction hypothesis, we conclude that
$\spotential{\Gamma}{P} - \potential{w}{\varnothing} \geqslant \ell_1$ and
$\spotential{\Delta,\typed{x}{\alpha}}{R} - \potential{v}{Q'} \geqslant \ell_2$, from which the sub-case follows.

Second, we consider the more involved sub-case, where $e_1$ is of $\Tree$ type. Thus, wlog.~$\Xi$ ends
in the following application of the \rulelettreecf-rule.
\begin{equation*}
\infer[\rulelettreecf]{
  \tjudge{\Gamma, \Delta}{Q}{\flstk{let }x\flst{ = }e_1\flstk{ in }e_2}{\beta}{Q'}
}{%
  \begin{minipage}[b]{82ex}
    \centering
    $p_i = q_i$ \quad 
    $p_{(\veca,c)} = q_{(\veca, \vec{0}, c)}$ \quad
    $r_j = q_{m+j}$ \quad $r_{k+1} = p'_\ast$ \quad
    $r_{(\vec{0},d,e)} = p'_{(d,e)}$ \quad
    $\forall \vec{b} \ne \vec{0} \left( r_{(\vec{b},0,0)} = q_{(\vec{0},\vec{b},0)} \right)$
    \\[1ex]
    $\forall \vec{b} \ne \vec{0}, \vec{a} \ne \vec{0} \lor c \ne 0\
    \left( q_{(\veca,\vecb,c)} = \sum_{(d,e)} p^{(\vecb,d,e)}_{(\veca,c)} \right)$
    \\[1ex]
    $\forall \vec{b} \ne \vec{0}, d \ne 0 \lor e \neq 0 \
    \left(
    r_{(\vec{b},d,e)} = {p'}^{(\vecb,d,e)}_{(d,e)}
    \wedge
    \forall (d',e') \neq (d,e) \ \left(
    {p'}^{(\vecb,d,e)}_{(d',e')} = 0 \right) \wedge {} \right.$\\
    $\left. {} \land 
    \sum_{(a,c)} p^{(\vecb,d,e)}_{(\veca,c)} \ge
    {p'}^{(\vecb,d,e)}_{(d,e)} \land 
    \forall \veca \ne \vec{0} \lor c \neq 0 \
    \left( p^{(\vecb,d,e)}_{(\veca,c)} \neq 0 \rightarrow  {p'}^{(\vecb,d,e)}_{(d,e)} \leqslant p^{(\vecb,d,e)}_{(\veca,c)} \right) \right)$\\[2ex]
    $\tjudge{\Gamma}{P}{e_1}{\TreeShort}{P'}$
    \hfill
    $\forall {\vecb \ne \vec{0},d \ne 0 \lor e \neq 0} \ \left( \tjudgecf{\Gamma}{P^{(\vecb,d,e)}}{e_1}{\TreeShort}{{P'}^{(\vecb,d,e)}} \right)$
    \hfill
    $\tjudge{\Delta, \typed{x}{\TreeShort}}{R}{e_2}{\beta}{Q'}$
  \end{minipage}
}
\tkom
\end{equation*}
where the annotations $Q$, $P$, $R$, $Q'$, $P'$ and the
sequences $\veca$, $\vecb$ are as above.
Further, for each sequence $\vecb \not= \vec{0}$, $P^{(\vecb,u,v)}$ and ${P'}^{(\vecb,u,v)}$ denote annotations of length $m$.
By definition and due to the constraints expressed in the typing rule, we have for
all $\vecb \not= \vec{0}$:
\begin{align*}
  \spotential{\Gamma,\Delta}{Q} & =
  \sum_i q_i \rk(t_i) + \sum_j q_j \rk(u_j) +
  \sum_{\vec{a} \ne \vec{0} \lor \vec{b} \ne \vec{0} \lor c \ne 0} q_{(\veca,\vecb,c)} \log(\veca\size{\vect} + \vecb\size{\vecu} + c)
  \\
  \spotential{\Gamma}{P} & = \sum_i q_i \rk(t_i) + \sum_{\vec{a} \ne \vec{0} \lor c \ne 0} q_{(\veca, \vec{0}, c)} \log (\veca \size{\vect} + c) - K
  \\
  \potential{w}{P'} &= r_{k+1} \rk(w) + \sum_{a \ne 0 \lor c \neq 0} r_{(\vec{0},a,c)} \log (a\size{w} +c)
  \\
  \spotential{\Gamma}{P^{(\vecb,d,e)}} &= \sum_{\vec{a} \ne \vec{0} \lor c \ne 0} p^{(\vecb,d,e)}_{(\veca,c)} \log (\veca \size{\vect} + c)
  \\
  \potential{w}{{P'}^{(\vecb,d,e)}} &= {p'}^{(\vecb,d,e)}_{(d,e)} \log (d\size{w} + e)
  \\
  \spotential{\Delta,\typed{x}{\TreeShort}}{R} & = \sum_{j} q_j \rk(u_j) + r_{k+1} \rk(w) +
\sum_{\vec{a} \ne \vec{0} \lor d \ne 0 \lor e \ne 0} r_{(\vecb,d,e)} \log(\vecb\size{\vec{u}}+d\size{w}+e) + K
  \tkom
\end{align*}
where we set $\vect \defsym \seq[1]{t}[m]$ and $\vec{u} \defsym \seq[1]{u}[k]$, denoting the substitution instances of the variables in $\Gamma$, $\Delta$, respectively.

By main induction hypothesis, we conclude that $\spotential{\Gamma}{P} - \potential{w}{P'} \geqslant \ell_1$.
Further, for all $\vecb \not= \vec{0},d \ne 0 \lor e \neq 0$, we have, due to the cost-free typing constraints $\spotential{\Gamma}{P^{(\vecb,d,e)}} \geqslant \potential{w}{{P'}^{(\vecb,d,e)}}$.
The latter yields more succinctly (for all $\vecb \not= \vec{0},d \ne 0 \lor e \neq 0$) that
\begin{equation}
  \label{eq:soundness-alt}
  \sum_{\veca,c} p^{(\vecb,d,e)}_{(\veca,c)} \log (\veca \size{\vect} + c) \geqslant {p'}^{(\vecb,d,e)}_{(d,e)} \log (d\size{w} + e)
  \tpkt
\end{equation}
A third application of MIH yields that
$\spotential{\Delta,\typed{x}{\TreeShort}}{R} - \potential{v}{Q'} \geqslant \ell_2$.
Due to the conditions
$\sum_{(a,c)} p^{(\vecb,d,e)}_{(\veca,c)} \ge {p'}^{(\vecb,d,e)}_{(d,e)}$,
for all $(d',e') \neq (d,e)$, ${p'}^{(\vecb,d,e)}_{(d',e')} = 0$
and for all $\veca$, $c$
$\left( p^{(\vecb,d,e)}_{(\veca,c)} \neq 0 \rightarrow  {p'}^{(\vecb,d,e)}_{(d,e)} \leqslant p^{(\vecb,d,e)}_{(\veca,c)} \right)$,
we can apply Lemma~\ref{lem:log-inequality} to Equation~\eqref{eq:soundness-alt} and obtain
\begin{equation*}
  \sum_{\vec{a} \ne \vec{0} \lor c \ne 0} p^{(\vecb,d,e)}_{(\veca,c)} \log (\veca \size{\vect} + \vecb \size{\vecu} + c) \geqslant
  {p'}^{(\vecb,d,e)}_{(d,e)} \log (\vecb \size{\vecu} + d\size{w} + e)
  \tpkt
\end{equation*}
Due to the condition $\left( q_{(\veca,\vecb,c)} = \sum_{(d,e)} p^{(\vecb,d,e)}_{(\veca,c)} \right)$
for all $\vec{b} \ne \vec{0}, \vec{a} \ne \vec{0} \lor c \ne 0$, 
we can sum those equations for all $d \ne 0 \lor e \neq 0$ and obtain (for all $\vecb \not= \vec{0},d \ne 0 \lor e \neq 0$) that
\begin{equation*}
  \sum_{\vec{a} \ne \vec{0} \lor c \ne 0} q_{(\veca,\vecb,c)} \log (\veca \size{\vect} + \vecb \size{\vecu} + c) \geqslant
  \sum_{d \ne 0 \lor e \neq 0} r_{(\vec{b},d,e)} \log (\vecb \size{\vecu} + d\size{w} + e)
  \tpkt
\end{equation*}

We can combine the above fact to conclude the case.

\emph{Case.} Finally, we consider the application rules~\ruleapp\ and~\ruleappcf. As the cost-free variant only
differs in sofar that costs are not counted by~\ruleappcf, it suffices to consider the rule~\ruleapp.
Let $f(\seq{x}[k]) = e \in \Program$ and let $\Pi$ derives $\eval{\sigma}{\ell+1}{f(\seq{x}[k])}{v}$.
We consider the costed typing
$\tjudgecf{\typed{x_1}{\alpha_1},\dots,\typed{x_k}{\alpha_k}}{(P + K \cdot Q)+1}{f(\seq{x}[k])}{\alpha}{P' -1 + K \cdot Q'}$,
where $K \in \Qplus$. Set $\Gamma \defsym \typed{x_1}{\alpha_1},\dots,\typed{x_k}{\alpha_k}$.
As $\Program$ is well-typed, we have
\begin{equation*}
  \tjudge{\Gamma}{P}{e}{\beta}{P'} \qquad \text{and} \qquad \tjudgecf{\Gamma}{Q}{e}{\beta}{Q'}\tpkt
\end{equation*}
We can apply MIH wrt.\ the evaluation $\Pi'$ of $\eval{\sigma}{\ell}{e}{v}$
to conclude $\spotential{\Gamma}{P} - \potential{v}{P'} \geqslant \ell$ as well as
$\spotential{\Gamma}{Q} \geqslant \potential{v}{Q'}$.
By monotonicity of addition and multiplication:
\begin{align*}
  \spotential{\Gamma}{P + K \cdot Q} &= \spotential{\Gamma}{P} + K \cdot \spotential{\Gamma}{Q}\\
  & \geqslant (\potential{v}{P'} + \ell) + K \cdot \potential{v}{Q'} = \potential{v}{P' + K \cdot Q'} + \ell \tpkt
\end{align*}
Thus
\begin{gather*}
  \spotential{\Gamma}{P + K \cdot Q} - \potential{v}{P' - 1 + K \cdot Q'} = {} \\
  {} = (\spotential{\Gamma}{P + K \cdot Q} - \potential{v}{P' + K \cdot Q'}) +1 \geqslant \ell + 1
  \tpkt
\end{gather*}
From this, the case follows, which completes the proof of the soundness theorem.
\end{proof}


\begin{remark}
  We note that the basic resource functions can be generalised
  to additionally represent linear functions in the size of the arguments. The above soundness theorem is
  not affected by this generalisation.
\end{remark}

In the next section, we exemplify the use of the proposed type-and-effect system, cf.~Figure~\ref{fig:5}, on the motivating example.

\section{Analysis}
\label{Analysis}

As promised, we apply in this section the proposed type-and-effect system to obtain an optimal analysis of
the amortised costs of the zig-zig case of \emph{splaying}, once the type annotations are fixed. As a preparatory
step, also to emphasise the need for the cost-free semantics, we precise the informal account of
compositional reasoning given in Section~\ref{Primer}.

\subsection{Let-Normal Form}
\label{Analysis:1}

We consider the expression \tree{al}{a}{\tree{ar}{b}{\tree{br}{c}{cr}}} =: t', which becomes the following expression $e$ in let-normal form
\begin{lstlisting}
  let t''' = (br, c, cr) in (
      let t'' = (ar, b, t''') in ((al, a, t''))
  )
\end{lstlisting}
The expression $e$ is typable with the following derivation, where the expression $e'$ abbreviate \lstinline{let t'' = (ar, b, t''') in ((al, a, t''))}.
(We have ignored expressions of base type to increase readability.) 
\begin{equation*}
  \infer[(\ast)]{\tjudge{\typed{br}{\TreeShort}, \typed{cr}{\TreeShort}, \typed{ar}{\TreeShort}, \typed{al}{\TreeShort}}{Q}{e}{\TreeShort}{Q'}}{%
      \tjudge{\typed{br}{\TreeShort},\typed{cr}{\TreeShort}}{Q_1}{(br,c,cr)}{\TreeShort}{Q'_1}
      &
      \infer{\tjudge{\typed{ar}{\TreeShort},\typed{al}{\TreeShort},\typed{t'''}{\TreeShort}}{Q_2}{e'}{\TreeShort}{Q'}}{%
        \infer{\tjudge{\typed{ar}{\TreeShort}, \typed{t'''}{\TreeShort}}{Q_3}{(ar,b,t''')}{\TreeShort}{Q'_3}}{%
            \begin{minipage}{30ex}
              $q^3_1 = q^3_2 = {q'}^3_\ast$\\
              $q^3_{(1,0,0)} = q^3_{(0,1,0)} = {q'}^3_\ast$\\
              $q^3_{(1,1,0)} = {q'}^3_{(1,0)}$
            \end{minipage}
          }
          &
          \qquad \eqref{eq:subtree}
        }        
    }
\end{equation*}
Here, we employ the derivability of the following type judgement~\eqref{eq:subtree}
by a single application of~\rulenode, wrt.\ the annotation $Q_4$, $Q'$ given below.
\begin{equation}
  \label{eq:subtree}
  \tjudge{\typed{al}{\TreeShort}, \typed{t''}{\TreeShort}}{Q_4}{(al,a,t'')}{\TreeShort}{Q'}
  \tpkt
\end{equation}

It is not difficult to check, that the above derivation indeed proves well-typedness
of the expression $e$ wrt. the below given type annotations.
\begin{align*}
  Q &\colon q_1 = q_2 = q_3 = q_4 = 1; \RED{q_{(1,1,1,0,0)} = 1}; q_{(1,1,0,0,0)} = 1; q_{(0,0,0,0,2)} = 1 \tkom
  \\
    & \quad q_{(1,0,0,0,0)} = q_{(0,1,0,0,0)} = q_{(0,0,1,0,0)} = q_{(0,0,0,1,0)} = 1 \tkom
  \\
  Q' &\colon q'_\ast = 1, q'_{(0,2)} = 1 \tkom
  \\
  Q_1 &\colon q^1_1 = q_1=1; q^1_2 = q_2=1; q^1_{(0,0,2)} = q_{(0,0,0,0,2)} = 1 \tkom
  \\
    & \quad q^1_{(1,1,0)} = q_{(1,1,0,0,0)}; q^1_{(1,0,0)} = q_{(1,0,0,0,0)} = 1;
      q^1_{(0,1,0)} =  q_{(0,1,0,0,0)} = 1 \tkom
  \\
  Q'_1 &\colon {q'}^1_\ast = 1; {q'}^1_{(1,0)} = 1; {q'}^1_{(0,2)} = 1 \tkom
  \\
  P^{(1,0,1,0)} & \colon \RED{{p}^{(1,0,1,0)}_{(1,1,0)} = q_{(1,1,1,0,0)} = 1} \tkom
  \\
  {P'}^{(1,0,1,0)} & \colon \RED{{p'}^{(1,0,1,0)}_{(1,0)} = 1} \tkom
  \\
  Q_2 &\colon q^2_1 = q_3=1; q^2_2 = q_4=1; q^2_3 = {q'}^1_\ast = 1; \tkom
  \\
  & \quad {q}^2_{(1,0,0,0)} = q_{(0,0,1,0,0)}=1; {q}^2_{(0,1,0,0)} = q_{(0,0,0,1,0)}=1;
        \tkom
  \\
  & \quad {q}^2_{(0,0,1,0)} = {q'}^1_{(1,0)} = 1; {q}^2_{(0,0,0,2)} = {q'}^1_{(0,2)}=1;
    \RED{q^2_{(1,0,1,0)} = {p'}^{(1,0,1,0)}_{(1,0)} = 1} \tkom
  \\
  Q_3 &\colon q^3_1 = q^2_1=1; q^3_2 = q^2_3=1; q^3_{(0,0,2)}={q}^2_{(0,0,0,2)}=1
  \\
    & \quad q^3_{(1,0,0)} = {q^2}_{(1,0,0,0)}=1; q^3_{(0,1,0)} = {q^2}_{(0,0,1,0)}=1;
      q^3_{(1,1,0)} = {q^2}_{(1,0,1,0)}=1 \tkom
  \\
  Q_3' &\colon {q'}^3_\ast=1, {q'}^3_{(1,0)}=1, {q'}^3_{(0,2)}=1
  \\
  Q_4  &\colon q^4_1 = q^2_2=1; q^4_{(0,0,2)}={q'}^3_{(0,2)}=1;
         q^4_{(1,0,0)} = {q}^2_{(0,1,0,0)} = 1;
         q^4_{(0,1,0)} = {q'}^3_{(1,0)} = 1
\end{align*}

In the inference marked with $(\ast)$, we employ the (almost trivial) correctness of the
following \emph{cost-free} typing derivation for
$\tjudgecf{\typed{br}{\TreeShort}, \typed{cr}{\TreeShort}}{P^{(1,0,1,0)}}{(br,c,cr)}{\TreeShort}{{P'}^{(1,0,1,0)}}$.
(For instantiation of the rule~\rulelettreecf\, note $\vec{b}=(1,0)$.)
\begin{equation}
  \label{eq:costfree}
  \infer{\tjudgecf{\typed{br}{\TreeShort}, \typed{cr}{\TreeShort}}{P^{(1,0,1,0)}}{(br,c,cr)}{\TreeShort}{{P'}^{(1,0,1,0)}}}{%
    p^{(1,0,1,0)}_{(1,1,0)} = {p'}^{(1,0,1,0)}_{(1,0)}
  }
  \tpkt
\end{equation}
For all $\vec{b} \not= (0)$, $\vec{b} \not= (1)$ and arbitrary $d$, $e$,
we set $P^{(\vec{b},d,e)} = {P'}^{(\vec{b},d,e)} \defsym \varnothing$.
Our prototype fully automatically checks correctness of the above given annotations.

We emphasise that the involved \rulelet-rule, employed in step $(\ast)$ cannot be avoided.
In particular, the additional cost-free derivation~\eqref{eq:costfree} is essential. Observe the annotation marked
in red in the calculation above. Eventually these amount to a shared potential employed in step ($\ast$).
The cost-free semantics allows us to exploit this shared potential, which otherwise would have to be discarded.

To wit, assume momentarily the rule~\rulelet\ would not make use of cost-free reasoning, similar to
the simplified \rulelet-rule, that we have used in the explanations on page~\pageref{page:letrule}. Then the
shared potential represented by the coefficient $q_{(1,1,1,0,0)} \in Q$ is discarded by the rule. However, this potential
is then missing, if we attempt to type the judgement 
\begin{equation*}
  \tjudge{\typed{ar}{\TreeShort},\typed{al}{\TreeShort},\typed{t'''}{\TreeShort}}{R}{e'}{\TreeShort}{Q'}
  \tkom
\end{equation*}
where $R$ is defined as $Q_2$, except that $r_{(1,0,1,0)} = 0$. Thus, this attempt fails.
(Note that the corresponding coefficient of $Q_2$, marked in red, is non-zero.)

\begin{remark}
To some extent this is in contrast to the use of cost-free semantics in
the literature~\cite{Hoffmann:2011,HAH:2012b,HM:2015,HDW:2017,MS:2020}. While cost-free semantics appear as an add-on in these works, essential only if we
want to capture non tail-recursive programs, cost-free semantics is essential in our context---it is already required for the representation of simple values.
\end{remark}

\subsection{Splay Trees}
\label{Analysis:2}

\begin{figure}[h]
\centering
\begin{equation*}
\infer={\tjudge{\typed{a}{\BaseShort},\typed{t}{\TreeShort}}{Q}{\match\ t\
    \with \text{\lstinline{|}} \leaf\  \arrow\ \leaf
    \text{\lstinline{|}} \flsttree{cl}{c}{cr}\ \arrow\ e_1}{\TreeShort}{Q'}}{%
  \infer{\tjudge{\typed{a}{\BaseShort},\typed{cl}{\TreeShort},\typed{c}{\BaseShort},\typed{cr}{\TreeShort}}{Q_1}{\cif\ a=c\
    \cthen\ \flsttree{cl}{c}{cr}\ \celse\ e_2}{\TreeShort}{Q'}}{%
  \infer={%
     \tjudge{\typed{a}{\BaseShort},\typed{b}{\BaseShort},\typed{cl}{\TreeShort},\typed{cr}{\TreeShort}}{Q_1}{\match\ cl\ \with\
     \text{\lstinline{|}} \leaf\ \arrow\ \flsttree{cl}{c}{cr}
     \text{\lstinline{|}} \flsttree{bl}{b}{br}\ \arrow\ e_3}{\TreeShort}{Q'}}{%
     \infer[\rulew]{\tjudge{\Gamma,\typed{cr}{\TreeShort},\typed{bl}{\TreeShort},\typed{br}{\TreeShort}}{Q_2}{e_3}{\TreeShort}{Q'}}{%
     \infer={\tjudge{\Gamma,\typed{cr}{\TreeShort},\typed{bl}{\TreeShort},\typed{br}{\TreeShort}}{Q_3}{e_3}{\TreeShort}{Q'}}{%
       \infer[(\ast)]{\tjudge{\Gamma,\typed{cr}{\TreeShort},\typed{bl}{\TreeShort},\typed{br}{\TreeShort}}{Q_3}{e_4}{\TreeShort}{Q'}}{%
         \infer{\tjudge{\typed{a}{\BaseShort},\typed{bl}{\TreeShort}}{Q}{\text{\lstinline{splay a bl}}}{\TreeShort}{Q'-1}}{%
           \text{\lstinline{splay:}}{ \atypdcl{\TreeShort}{Q}{\TreeShort}{Q'}}}
         &
         \infer={\tjudge{\Delta,\typed{cr}{\TreeShort},\typed{br}{\TreeShort},\typed{x}{\TreeShort}}{Q_4}{\match\ x\ \with\
             \text{\lstinline{|}} \flsttree{al}{a'}{ar}\ \arrow t'}{\TreeShort}{Q'}}{%
            \tjudge{\Delta,\typed{cr}{\TreeShort},\typed{br}{\TreeShort},\typed{al}{\TreeShort},\typed{a'}{\BaseShort},\typed{ar}{\TreeShort}}{Q_5}{t'}{\TreeShort}{Q'}
         }
       }
     }
   }
 }
}
}
\end{equation*}
\caption{Partial Typing Derivation for \lstinline{splay}, focusing on the zig-zig Case.}
\label{fig:6}
\end{figure}

\begin{figure}[h]
\centering
\begin{equation*}
\infer={\tjudge{\typed{a}{\BaseShort},\typed{t}{\TreeShort}}{P}{\match\ t\
    \with \text{\lstinline{|}} \leaf\  \arrow\ \leaf
    \text{\lstinline{|}} \flsttree{cl}{c}{cr}\ \arrow\ e_1}{\TreeShort}{P'}}{%
  \infer{\tjudge{\typed{a}{\BaseShort},\typed{cl}{\TreeShort},\typed{c}{\BaseShort},\typed{cr}{\TreeShort}}{P_1}{\cif\ a=c\
    \cthen\ \flsttree{cl}{c}{cr}\ \celse\ e_2}{\TreeShort}{P'}}{%
  \infer={%
     \tjudge{\typed{a}{\BaseShort},\typed{b}{\BaseShort},\typed{cl}{\TreeShort},\typed{cr}{\TreeShort}}{P_1}{\match\ cl\ \with\
     \text{\lstinline{|}} \leaf\ \arrow\ \flsttree{cl}{c}{cr}
     \text{\lstinline{|}} \flsttree{bl}{b}{br}\ \arrow\ e_3}{\TreeShort}{P'}}{%
      \infer={\tjudge{\Gamma,\typed{cr}{\TreeShort},\typed{bl}{\TreeShort},\typed{br}{\TreeShort}}{P_2}{e_3}{\TreeShort}{P'}}{%
       \infer[(\ast)]{\tjudge{\Gamma,\typed{cr}{\TreeShort},\typed{bl}{\TreeShort},\typed{br}{\TreeShort}}{P_2}{e_4}{\TreeShort}{P'}}{%
         \infer{\tjudge{\typed{a}{\BaseShort},\typed{bl}{\TreeShort}}{\varnothing}{\text{\lstinline{splay a bl}}}{\TreeShort}{\varnothing}}{%
           \text{\lstinline{splay:}}{ \atypdcl{\TreeShort}{\varnothing}{\TreeShort}{\varnothing}}}
         &
         \infer={\tjudge{\Delta,\typed{cr}{\TreeShort},\typed{br}{\TreeShort},\typed{x}{\TreeShort}}{P_3}{\match\ x\ \with\
             \text{\lstinline{|}} \flsttree{al}{a'}{ar}\ \arrow t'}{\TreeShort}{P'}}{%
            \tjudge{\Delta,\typed{cr}{\TreeShort},\typed{br}{\TreeShort},\typed{al}{\TreeShort},\typed{a'}{\BaseShort},\typed{ar}{\TreeShort}}{P_4}{t'}{\TreeShort}{P'}
         }
       }
     }
   }
 }
}
\end{equation*}
\caption{Cost-Free Derivation for \lstinline{splay}, focusing on the zig-zig Case.}
\label{fig:7}
\end{figure}

In this subsection, we exemplify the use of the type system presented in the last section
on the function \lstinline{splay}, cf.~Figure~\ref{fig:1}.
Our amortised analysis of splaying yields that the amortised cost
of \lstinline{splay a t} is bound by $3\log(\size{t})+1$, where the actual
cost counts the number of recursive calls to \lstinline{splay}, cf.~\cite{ST:1985,Schoenmakers93,Nipkow:2015}.
To verify this amortised cost, we derive
\begin{equation}
  \label{eq:splay:1}
  \tjudge{\typed{a}{\Base},\typed{t}{\Tree}}{Q}{e}{\Tree}{Q'} \tkom
\end{equation}
where the expression $e$ is the definition of \lstinline{splay} given in
Figure~\ref{fig:1} and the annotations $Q$ and $Q'$ are as follows:
\begin{align*}
  Q &\colon q_1=1, q_{(1,0)} = 3, q_{(0,2)} = 1 \tkom
  \\
  Q' &\colon q'_\ast =1 \tpkt  
\end{align*}
Remark that the amortised cost of splaying is represented by the coeficients
$q_{(1,0)}$ and $q_{(0,2)}$, expressing in sum $3\log(\size{t})+1$. Note, further
that the coefficient $q_1, q'_\ast$, represent Schoenmakers' potential,
that is, $\rk(t)$ and $\rk(\text{\lstinline{splay a t}})$, respectively. 

We restrict to the zig-zig case: \lstinline{t = } \tree{\tree{bl}{b}{br}}{c}{cr}
together with the recursive call \lstinline{splay a bl} = \flsttree{al}{a'}{ar} and
$a < b < c$. Thus \lstinline{splay a t} yields \tree{al}{a'}{\tree{ar}{b}{\tree{br}{c}{cr}}} =: t'.
Recall that $a$ need not occur in $t$, in this case, the last element $a'$ before
a leaf was found, is rotated to the root.
Our prototype checks correctness of these annotations automatically.

Let $e_1$ denote the subexpression of the definition of splaying, starting in
program line $4$. On the other hand let $e_2$ denote the subexpression defined from
line $5$ to $15$ and let $e_3$ denote the program code within $e_2$ starting
in line $8$. Finally the expression in lines $11$ and $12$,
expands to the following, if we remove part of the syntactic sugar:
\begin{equation*}
  e_4 \defsym \vlet\ x\ \equal\ \text{\lstinline{splay a bl}}\ \vin\
    \match\ x\ \with\ \text{\lstinline{|}} \leaf\ \arrow \leaf\
    \text{\lstinline{|}} \flsttree{al}{a'}{ar}\ \arrow t'
    \tpkt
\end{equation*}

Figure~\ref{fig:6} shows a simplified derivation of~\eqref{eq:splay:1}, where
we  have focused only on a particular path in the derivation tree, suited
to the considered zig-zig case of the definition of splaying.
Omission of premises or rules is indicated by double lines in the inference step.
Again we make crucial use of the cost-free semantics in this derivartion. The corresponding inference step in Figure~\ref{fig:6} 
is marked with ($\ast$) and the employed shared potentials are marked in red.

We abbreviate $\Gamma \defsym \typed{a}{\BaseShort},\typed{b}{\BaseShort},\typed{c}{\BaseShort}$,
$\Delta \defsym \typed{b}{\BaseShort},\typed{c}{\BaseShort}$.
In addition to the original signature of splaying, $\atypdcl{\BaseShort \times \TreeShort}{Q}{\TreeShort}{Q'}$, we use the following annotations, induced by constraints in the type system, cf.~Figure~\ref{fig:5}. As
in Section~\ref{Analysis:1}, we mark annotations that require cost-free derivations in the \rulelettreecf\ rule in red.
\begin{align*}
  Q_1 &\colon q^1_1 = q^1_2 = q_1 = 1, q^1_{(1,1,0)} = q_{(1,0)} = 3,
      q^1_{(1,0,0)} = q^1_{(0,1,0)} = q_1 = 1,
      q^1_{(0,0,2)} = q_{(0,2)} = 1 \tkom
  \\
  Q_2 &\colon q^2_1 = q^2_2 = q^2_3 = 1,
        q^2_{(0,0,0,2)} = 1,
        q^2_{(1,1,1,0)} = q^1_{(1,1,0)} = 3,
        q^2_{(0,1,1,0)} = q^1_{(1,0,0)} = 1,
  \\
      & \quad
        q^2_{(1,0,0,0)} = q^1_{(0,1,0)} = 1,
        q^2_{(0,1,0,0)} = q^2_{(0,0,1,0)} = q^1_1 = 1 \tkom
  \\
  Q_3 &\colon q^3_1=q^3_2=q^3_3 =1,
        q^3_{(0,0,0,2)} = 2, \\
      & \quad
        q^3_{(0,1,0,0)} = 3,  q^3_{(1,0,0,0)} = q^3_{(0,0,1,0)} = q^3_{(1,0,1,0)} = \RED{q^3_{(1,1,1,0)}} =1 \tpkt
\end{align*}

In the step marked with the rule \rulew\ in Figure~\ref{fig:6}, a~\emph{weakening} step is applied, which
amounts to the following inequality:
\begin{equation*}
  \potential{\Gamma,\typed{cr}{\TreeShort},\typed{bl}{\TreeShort},\typed{br}{\TreeShort}}{Q_2} \geqslant
  \potential{\Gamma,\typed{cr}{\TreeShort},\typed{bl}{\TreeShort},\typed{br}{\TreeShort}}{Q_3} \tpkt
\end{equation*}
We emphasise that this step can neither be avoided, nor easily moved to the axioms of the derivation.
To verify the correctness of \emph{weakening} through a direct comparison.
Let $\sigma$ be a substitution. Then, we have
\begin{align*}
  \spotential{\typed{cr}{\TreeShort},\typed{bl}{\TreeShort},\typed{br}{\TreeShort}}{Q_2} & =
  1 + \rk(cr)+\rk(bl)+\rk(br) + 3\log(\size{cr} + \size{bl} + \size{br}) + {} \\
  & \quad {} +    \log(\size{bl}+\size{br}) + \log(\size{cr}) + \log(\size{bl}) + \log(\size{br})
  \\[1ex]
  & =
  1 + \rk(cr)+\rk(bl)+\rk(br) + 2\log(\size{t}) + \log(\size{t}) + {}
  \\
  & \quad {} + \log(\size{bl}+\size{br}) + \log(\size{cr}) + \log(\size{bl}) + \log(\size{br})
  \\[1ex]
  & \geqslant
  1 + \rk(cr)+\rk(bl)+\rk(br) + \log(\size{bl}) + \log(\size{br}+\size{cr}) + 2 + {}
  \\
  & \quad {} + \log(\size{bl}+\size{br}+\size{cr}) + {}
  \\
  & \quad {} + \log(\size{bl}+\size{br}) + \log(\size{cr}) + \log(\size{bl}) +\log(\size{br}) + {}
  \\[1ex]
  & \geqslant \rk(bl) + 1+ 3\log(\size{bl}) + \rk(cr)+\rk(br)+ \log(\size{br}) + {}
  \\
  & \quad {} + \log(\size{cr}) + \log(\size{br}+\size{cr}) + {}
  \\
  & \quad {} + \log(\size{bl}+\size{br}+\size{cr}) + 1
  =  \spotential{\typed{cr}{\TreeShort},\typed{bl}{\TreeShort},\typed{br}{\TreeShort}}{Q_3}
  \tpkt
\end{align*}
Note that we have used Lemma~\ref{l:1} in the third line to conclude
\begin{equation*}
  2 \log(\size{t}) \geqslant \log(\size{bl}) + \log(\size{br}+\size{cr}) + 2 \tkom
\end{equation*}
as we have $\size{t} = \size{\tree{\tree{bl}{b}{br}}{c}{cr}} = \size{bl} + \size{br} + \size{cr}$.
Furthermore, we have only used monotonicity of $\log$ and formal simplifications.

Further, we verify the use of the \rulelettreecf-rule, marked with $(\ast)$ in the proof.
Consider the following annotation $Q_4$:
\begin{align*}
  Q_4 & \colon q^4_1 = q^3_1; q^4_2 = q^3_3; q^4_3 = q'_\ast;
        q^4_{(1,0,0,0)} = q^3_{(1,0,0,0)};
        q^4_{(0,1,0,0)} = q^3_{(0,0,1,0)};
  \\
  & \quad      q^4_{(1,1,0,0)} = q^3_{(1,0,1,0)};
        q^4_{(1,1,1,0)} = {p'}^{(1,1,1,0)}_{(1,0)} = 1
  \\
  P^{(1,1,1,0)} & \colon \RED{p^{(1,1,1,0)}_{(1,0)}} = \RED{q^3_{(1,1,1,0)}} = \RED{1}
  \\
  {P'}^{(1,1,1,0)} & \colon \RED{{p'}^{(1,1,1,0)}_{(1,0)}} = \RED{1}
  \end{align*}
To see that $Q_4$ is consistent with the constraints on resource annotations in the \rulelettreecf-rule, we first note that
\begin{align*}
  Q+1 & \colon q= q^3_2 = 1, q_{(1,0)} = q^3_{(0,1,0,0)} = 3;
        q_{(0,2)} = q^3_{(0,0,0,2)} \tpkt
\end{align*}
Hence the constraints on the annotations for the left typing tree in the \rulelettreecf-rule amount to the following:
\begin{equation*}
  q = q_2^3 = 1 \quad q_{(1,0)} = q^3_{(0,1,0,0)} = 3 \quad q_{(0,2)} = q_{(0,0,0,2)} = 2 \quad
  q'_\ast = q_3^4 = 1 \tkom
\end{equation*}
which are fulfilled. Further, the right typing tree yields the constraints:
\begin{align*}
  & q^4_1 = q^3_1 = 1
    \quad q^4_2 = q^3_3 = 1
    \quad q^4_{(1,0,0,0)} = q^3_{(1,0,0,0)} = 1
    \quad q^4_{(0,1,0,0)} = q^3_{(0,0,1,0)} = 1
  \\
  & q^4_{(1,1,0,0)} = q^3_{(1,0,1,0)} = 1
    \tkom
\end{align*}
which are also fulfilled. Hence, it remains to check the correctness of the constraints
for the actual cost-free derivation. First, note that for the vector $\vecb = (1,1)$, the
cost-free derivation needs to be checked wrt.\ the annotation pair $P^{(1,1,1,0)} = [p^{(1,1,1,0)}_{(1,0)}]$
and ${P'}^{(1,1,1,0)} = [{p'}^{(1,1,1,0)}_{(1,0)}]$. Second, the various
constraints in the rule \rulelettreecf\ simplify to the inequality $p^{(1,1,1,0)}_{(1,0)} \geqslant {p'}^{(1,1,1,0)}_{(1,0)}$, which holds. Third, the actual cost-free type derivation reads as follows:
\begin{equation}
  \label{eq:cf}
  \tjudge{\typed{a}{\BaseShort},\typed{bl}{\TreeShort}}{P^{(1,1,1,0)}}{\text{\lstinline{splay a bl}}}{\TreeShort}{{P'}^{(1,1,1,0)}} \tpkt
\end{equation}
The typing judgement~\eqref{eq:cf} is derivable if the following cost-free signatures are employed for splaying:
\begin{equation*}
  \text{\lstinline{splay}} \colon \atypdclcf{\Tree}{P}{\Tree}{P'} \qquad \atypdclcf{\Tree}{\varnothing}{\Tree}{\varnothing}
  \tkom
\end{equation*}
where $P = [p_{(1,0)}]$, $P' = [p'_{(1,0)}]$, with $p_{(1,0)} = p'_{(1,0)} \defsym 1$. Recall that $\varnothing$ denotes the
empty annotation, where all coefficients are set to zero. By definition, $P = P^{(1,1,1,0)}$ and $P' = {P'}^{(1,1,1,0)}$.
Informally, this cost-free signature is admissible, as the following equality holds:
\begin{equation*}
  \spotential{\typed{a}{\BaseShort}, \typed{bl}{\TreeShort}}{P} = \log(\size{bl}) =
  \log(\size{(al,a',ar)}) = \potential{\typed{\tree{al}{a'}{ar}}{\TreeShort}}{P'}
  \tpkt
\end{equation*}
Recall that we have \lstinline{splay a bl} = \tree{al}{a'}{ar} for the recursive call and that
$\size{bl} = \size{\text{\tree{al}{a'}{ar}}}$.
Formally, the type derivation of~\eqref{eq:cf} proceeds similar to the derivation in Figure~\ref{fig:6} in
conjunction with the analysis in Subsection~\ref{Analysis:1}, see Figure~\ref{fig:7}. As indicated also the cost-free derivation
requires the use of the full version of the rule \rulelettreecf, as marked by $(\ast)$. In particular, the informal argument on the
size of the argument and the result of splaying is built into the type system. We
use the following annotations:
\begin{align*}
  P &\colon p_{(1,0)} = 1 & P' &\colon p'_{(1,0)} = 1 \\
  P_1 & \colon p^1_{(1,1,0)} = p_{(1,0)} = 1 \\
  P_2 & \colon \RED{p^2_{(1,1,1,0)} =  p^1_{(1,1,0)} = 1}\\
  P_3 & \colon p^3_{(1,1,1,0)} = p'_{(1,0)} = 1\\
  P_4 & \colon p^4_{(1,1,1,1,0)} = p^3_{(1,1,1,0)} = 1
\end{align*}

Finally, one further application of the \rulematch-rule, yields the desired
derivation for suitable~$Q_5$. See also the previous subsection. Note that one further application of the \emph{weakening} rule is required here.

\section{Linearisation and Expert Knowledge}
\label{Automation}

In the context of the presented type-and-effect system (see Figure~\ref{fig:5})
an obvious challenge is the requirement to compare potentials symbolically (see Section~\ref{Typesystem}) rather
than compare annotations directly. This is in contrast to results on resource analysis for constant amortised costs, see eg.~\cite{JLHSH09,JHLH10,HAH:2012b,HDW:2017,JVFH:2017}.
Furthermore, the presence of logarithmic basic functions seems to necessitate
the embodiment of nonlinear arithmetic. In particular, as we need to make use of basic laws of the $\log$ functions,
as the following one. A variant of the below fact has already been observed by Okasaki, cf.~\cite{Okasaki:1999}.

\begin{lemma}
\label{l:1}
Let $x, y \geqslant 1$. Then $2 + \log(x) + \log(y) \leqslant 2\log(x+y)$.
\end{lemma}
\begin{proof}
We observe
  \begin{equation*}
    (x+y)^2 -4xy = (x-y)^2 \geqslant 0
    \tpkt
  \end{equation*}
Hence $(x+y)^2 \geqslant 4xy$ and from the monotonicity
of $\log$ we conclude $\log(xy) \leqslant \log(\frac{(x+y)^2}{4})$.
By elementary laws of $\log$ we obtain:
\begin{equation*}
  \log\left(\frac{(x+y)^2}{4}\right) = \log\left( \left(\frac{x+y}{2}\right)^2 \right) =
  2 \log(x+y) - 2
  \tkom
\end{equation*}
from which the lemma follows as $\log(xy) = \log(x)+\log(y)$.
\end{proof}

A refined and efficient approach which targets linear constraints is achievable as follows.
All logarithmic terms, that is, terms of the form $\log(.)$ are replaced by new variables, focusing on finitely many.
For the latter, we exploit the condition that in resource annotations only finitely many coefficients are non-zero.
Consider the following inequality as prototypical example. Validity of the constraint ought to incorporate the monotonicity of $\log$.
\begin{equation}
  \label{eq:1}
  a_1 \log(\size{t}) + a_2 \log(\size{cr}) \geqslant  b_1 \log(\size{t}) + b_2 \log(\size{cr})
  \tkom
\end{equation}
where we assume $t = \tree{cl}{c}{cr}$ for some value $c$ and thus $\size{t} \geqslant \size{cr}$, cf.~Section~\ref{Analysis:2}.
Replacing $\log(\size{t}), \log(\size{cr})$ with new unknowns $x$, $y$, respectively, we represent~\eqref{eq:1} as follows:
\begin{equation}
  \label{eq:2}
  \forall x, y \geqslant 0. \ a_1 x + a_2 y \geqslant  b_1 x + b_2 y
  \tkom
\end{equation}
Here we keep the side-condition $x \geqslant y$ and observe that the unknowns $x$, $y$ can be assumed
to be non-negative, as they represent values of the $\log$ function.
Thus properties like e.g.\ monotonicity of $\log$, as well as properties like Lemma~\ref{l:1} above,
can be expressed as inequalities over the introduced unknowns. E.g., the inequality $x \geqslant y$
above represents the axiom of monotonicity $\log(\size{t}) \geqslant \log(\size{cr})$. All such obtained inequalities
are collected as \emph{expert or prior knowledge}.
This entails that~\eqref{eq:2} is equivalent to the following existential
constraint satisfaction problem:
\begin{equation}
 \label{eq:3}
  \exists c,d. \ a_1 \geqslant b_1 + c \land a_2 \geqslant d \land b_2 \leqslant c+d
  \tpkt
\end{equation}

We seek to systematise the derivation of inequalities such as~(\ref{eq:3}) from expert knowledge.
For that, we assume that the gathered prior knowledge is represented by a system of inequalities as $A\vec{x} \leqslant \vec{b}, \vec{x} \geqslant 0$, where $A$ denotes a matrix with as many rows as we have prior knowledge, $\vec{b}$ a column vector and $\vec{x}$ the column vector of unknowns of suitable length;
$\vec{x} \geqslant 0$ because $\log$ evaluates to non-negative values.

Below we discuss a general method for the derivation of inequalities such as~(\ref{eq:3}) based on the affine form of Farkas' lemma.
First, we state the variant of Farkas' lemma that we use in this article,
cf.~\cite{Schrijver:1999}. Note that $\vec{u}$ and $\vec{f}$ denote column vectors of suitable length.

\begin{lemma}[Farkas' lemma]
\label{l:farkas}
  Suppose  $A\vec{x} \leqslant \vec{b}, \vec{x} \geqslant 0 $ is solvable.
  Then the following assertions are
  equivalent.
  \begin{align}
    \forall \vec{x} \geqslant 0. \ A\vec{x} \leqslant \vec{b} \Rightarrow \vec{u}^T\vec{x} \leqslant \lambda
    \label{eq:farkas:1}
    \\
    \exists \vec{f} \geqslant 0. \  \vec{u}^T \leqslant \vec{f}^T A \land \vec{f}^T \vec{b} \leqslant \lambda
    \label{eq:farkas:2}
  \end{align}
\end{lemma}
\begin{proof}
  It is easy to see that from \eqref{eq:farkas:2}, we obtain \eqref{eq:farkas:1}. Assume~\eqref{eq:farkas:2}. Assume further that
  $A\vec{x} \leqslant \vec{b}$ for some column vector $\vec{x}$. Then we have
  \begin{equation*}
    \vec{u}^T\vec{x} \leqslant \vec{f}^T A \vec{x} \leqslant \vec{f}^T \vec{b} \leqslant \lambda
    \tpkt
  \end{equation*}
  Note that for this direction the assumption that $A\vec{x} \leqslant \vec{b}, \vec{x} \geqslant 0 $  is solvable is not required.

  With respect to the opposite direction, we assume~\eqref{eq:farkas:1}. By assumption, $A\vec{x} \leqslant \vec{b}, \vec{x} \geqslant 0 $
  is solvable. Hence, maximisation of $\vec{u}^T \vec{x}$ under the side condition $A\vec{x} \leqslant \vec{b}, \vec{x} \geqslant 0 $ is feasible. Let
  $w$ denote the maximal value.
  Due to~\eqref{eq:farkas:1}, we have $w \leqslant \lambda$.

  Now, consider the dual asymmetric
  linear program to minimise  $\vec{y}^T \vec{b}$ under side condition $\vec{y}^T A = \vec{u}^T$ and $\vec{y} \geqslant 0$. Due to
  the Dualisation Theorem, the dual problem is also solvable with the same solution
  \begin{equation*}
    \vec{y}^T \vec{b} = \vec{u}^T \vec{x} = w \tpkt
  \end{equation*}
  We define $\vec{f} \defsym \vec{y}$, which attains the optimal value $w$, such that $\vec{f}^T A = \vec{u}^T$ and $\vec{f} \geqslant 0$
  such that $\vec{f}^T \vec{b} = w \leqslant \lambda$. This yields~\eqref{eq:farkas:2}.
\end{proof}

Second, we discuss a method for the derivation of inequalities such as~(\ref{eq:3}) based on Farkas' lemma.
Our goal is to automatically discharge symbolic constraints such as $\potential{\Gamma}{P} \leqslant \potential{\Gamma}{Q}$---%
as well as $\potential{\Gamma}{P'} \geqslant \potential{\Gamma}{Q'}$---as required by the \emph{weakening} rule~\rulew\ (see Section~\ref{Typesystem}).

According to the above discussion we can represent the inequality $\potential{\Gamma}{P} \leqslant \potential{\Gamma}{Q}$ by
\begin{equation*}
  \vec{p}^T \vec{x} + c_p \leqslant \vec{q}^T \vec{x} + c_q
  \tkom
\end{equation*}
where $\vec{x}$ is a finite vector of variables representing the base potential functions,
$\vec{p}$ and $\vec{q}$ are column vectors representing the unknown coefficients of the non-constant potential functions,
and $c_p$ and $c_q$ are the coefficients of the constant potential functions.
We assume the expert knowledge is given by the constraints
$A\vec{x} \leqslant \vec{b}, \vec{x} \geqslant 0$.
We now want to derive conditions for $\vec{p}$, $\vec{q}$, $c_p$, and $c_q$ such that we can guarantee
\begin{equation}
 \label{eq:4}
 \forall \vec{x} \geqslant 0. \ A \vec{x} \leqslant \vec{b} \Rightarrow \vec{p}^T \vec{x} + c_p \leqslant \vec{q}^T \vec{x} + c_q \tpkt
\end{equation}
By Farkas' lemma it is sufficient to find coefficients $\vec{f} \geqslant 0$ such that
\begin{equation}
 \label{eq:5}
  \vec{p}^T \leqslant \vec{f}^T A + \vec{q}^T \land \vec{f}^T \vec{b} + c_p \leqslant c_q \tpkt
\end{equation}
Hence, we can ensure Equation~\eqref{eq:4} by Equation~\eqref{eq:5} using the new unknowns $\vec{f}$.

We illustrate Equation~\eqref{eq:5} on the above example.
We have $A = (-1 \ 1)$, $b = 0$, $\vec{p} = (b_1 \ b_2)^T$, $\vec{q} = (a_1 \ a_2)^T$ as
well as $c_p = c_q = 0$. Then, the inequality $f b + c_p \leqslant c_q$ simplifies to $0 \leqslant 0$ and can in the following be omitted.
With the new unknown $f \geqslant 0$ we have
\begin{equation*}
    (b_1 \ b_2) \leqslant f (-1 \ 1) + (a_1 \ a_2) \tkom
\end{equation*}
which we can rewrite to
\begin{equation*}
  b_1 \leqslant -f + a_1 \land  b_2 \leqslant f + a_2 \tkom
\end{equation*}
easily seen to be equivalent to Equation~\eqref{eq:3}.

Thus, the validity of constraints incorporating the monotonicity of log becomes expressible in a systematic way. Further, the symbolic constraints enforced
by the \emph{weakening} rule can be discharged systematically and become expressible as existential constraint satisfaction problems.
Note that the incorporation of Farkas' lemma in the above form subsumes the well-known practise of coefficient comparison for the synthesis of
polynomial interpretations~\cite{contejean:2005}, ranking functions~\cite{PR:2004} or resource annotations in the context
of constant amortised costs~\cite{HAH:2012b}.

 In the next section, we briefly detail our implementation of the established logarithmic amortised resource analysis, based on the observations in this section.

\section{Implementation}
\label{sec:implementation}

Based on the principal approach, delineated in Section~\ref{Primer}, we have provided a prototype implementation of
the logarithmic amortised resource analysis detailed above.
The prototype is capable of \emph{type checking} a given resource annotation and requires
user interaction to specify the structural inferences sharing and weakening. These can be
applied manually to improve efficiency of type checking. In future work, we will strive for full automation, capable of \emph{type inference}.
In this section, we briefly indicate the corresponding design choices and heuristics used. Further, we present restrictions and future development areas of the prototype developed.

\emph{Template potential functions.}
Our potential-based method employs linear combinations of basic potential functions~$\BF$, cf.~Definition~\ref{d:basicpotential}. In order
to fix the cardinality of the set of resource functions to be considered, we restrict the coefficients of the potential functions
$p_{(\seq{a}[m],b)}$. For the non-constant part, we demand that $a_i \in \{0,1\}$, while the coefficients $b$, representing the constant part are
restricted to $\{0,1,2\}$.
This restriction to a relative small set of basic potential functions, suitable controls the number of constraints generated for each
inference rule in the type-and-effect system.

\emph{Type-and-effect system.}
Following ideas of classical Hindley-Milner type inference, we collect for each node in the abstract syntax tree (AST) of
the given program the constraints given by the corresponding inference rules in the type system (see Figure~\ref{fig:5}).
As a pre-requisite, we restrict ourselves to three type annotations employed for each function symbol.
\begin{inparaenum}[(i)]
\item One indeterminate type annotation representing a function call with costs;
\item one indeterminate cost-free type annotation to represent a zero-cost call; and
\item one fixed cost-free annotation with the empty annotation that doesn't carry any potential.
\end{inparaenum}
These restrictions were sufficient to handle the zig-zig case of splaying. A larger, potentially infinite set of type annotations is conceivable,
as long as well-typedness is respected, cf.~Definition~\ref{d:welltyped}. As noted in the context of the analysis of constant amortised
complexity an enlarged set of type annotations may be even required to handle non-tail recursive programs, cf.~\cite{HAH:2012b,HDW:2017}.
The collected constraints on the type annotations are passed to an SMT solver, in our case the SMT solver~\texttt{z3}%
\footnote{See~\url{https://github.com/Z3Prover/z3/wiki}.}
and solved over the positive rational numbers. Here we can directly encode the equalities and inequalities of the constraints given in the type system. Due to the use of Farkas' lemma (Lemma~\ref{l:farkas}) only linear constraints are generated.

Our implementation is currently only supports type checking, taking user guidance into account
and thus is semi-automated.
While deriving constraints for the AST nodes of the program is straightforward (as there is only one type rule for every syntactic statement of our programming language), we currently require user interaction for the application of the structural rules.

\emph{Structure rules.}
The structural rules can in principle be applied at every AST node of the program under analysis.
However, they introduce additional variables and constraints and for performance reasons it is better to apply them sparingly.
For the \emph{sharing} rule we proceed as follows:
We recall that the sharing rule allows us to assume that the type system is linear.
In particular, we can assume that every variable occurs exactly once in the type context, which is exploited in the definition of the \emph{let} rules.
However, such an eager application of the sharing rule would directly yield to a size explosion in the number of constraints, as the generation of each fresh variables requires the generation of exponentially many annotations.
Hence, we only apply sharing only when strictly necessary.
Similar to the sharing rule~\ruleshare, variable weakening~\rulewvar\ is employed only when required.
In this way the typing context can be kept small.
This in turn reduces the number of constraints generated.

For the \emph{weakening} rule, we employ our novel methods for symbolically comparing logarithmic expressions,
which we discussed in Section~\ref{Automation}.
Because of our use of Farkas' Lemma, weakening introduces new unknown coefficients, which again
may result in a forbiddingly large search space. Thus, weakening steps are particularly costly.
For performance reasons, we need to control the size of the resulting constraint system. Currently,
we rely on the user to specify the number and place of the applications of the weakening rule. This
is achieved through the provision of suitable \emph{tactics} for type checking.
Note that the weakening rule may need to be applied in the middle of a type derivation, see for example the typing derivation for our motivating example in~Figure~\ref{fig:6}.
This contrasts to the literature where the weakening rule can typically be incorporated
into the axioms of the type system and thus dispensed with.
Perhaps a similar approach is possible in the context of logarithmic amortised resource analysis. But
our current understanding does not support this.

\emph{Expert Knowledge.}
In Section~\ref{Automation}, we propose the generation of a suitable matrix $A$ collecting the \emph{expert or prior knowledge} on such inequalities. 
In particular, wrt.\ Lemma~\ref{l:1}, generation of this expert knowledge is straightforward.
The corresponding inequality amounts to a line in the expert knowledge matrix $A$.
Wrt.\ monotonicity we have experimented with a dedicated size analysis based on a simple static analysis of the given AST, as well as exploitation of the type annotations directly.
For the latter, note that the coefficients in the basic potential functions $p_{(\seq{a}[m],b)}$ are reflected in the corresponding type annotations.
Hence comparison of these (unknown) coefficients allows a sufficient size comparison of the data-structures (ie.\ trees) used in the program at hand.

For now, our exploitation of the expert knowlege is restricted to the monotonicity of logarithm together with the simple mathematical fact about logarithms presented in Lemma~\ref{l:1}. To improve the efficiency and effectivity of the methodology, the following additions could be explored:
  \begin{inparaenum}[(i)]
    \item additional mathematical facts on the logarithm function;
    \item improvement of the dedicated size analysis supporting the applicabilty of monotonicty laws;
    \item incorporation of basic static analysis techniques, like the result of a reachability analysis, etc.
  \end{inparaenum}

\section{Related Work}
\label{Related}

To the best of our knowledge the established type-and-effect system for the analysis of logarithmic amortised
complexity is novel and also the semi-automated resource analysis of self-balanced data-structures like splay trees, is
unparalleled in the literature.
However, there is a vast amount of literature on (automated) resource analysis.
Without hope for a completeness, we
briefly mention~\cite{AAGP08,ADFG:2010,conf/lpar/BlancHHK10,conf/pldi/GulwaniZ10,AAGP:2011,conf/sas/Alonso-BlasG12,HBCLMMP:2012,ADM15,avanzini2016tct,Flores-Montoya:2017,Giesl:2017} for an overview of the field.

(Constant) amortised cost analysis has been in particular
pioneered by Martin Hofmann and his collaborators.
Starting with seminal work on the static prediction of heap space
usage~\cite{HofmannJ03,HR:2013}, the approach has been
generalised to (lazy) functional programming~\cite{JLHSH09,JHLH10,HAH11,HAH:2012,HAH:2012b} and rewriting~\cite{HM:2014,HM:2015}.
Automation of amortised resource analysis has also been greatly influenced
by Hofmann, yielding to sophisticated tools for the analysis
of higher-order functional programs~\cite{HoffmannH10a,HoffmannH10b,Hoffmann:2011}, as well as of object-oriented programs~\cite{HR:2013,conf/lpar/BauerJH18}.
We mention here the highly sophisticated analysis behind
the \raml\ prototype developed in~\cite{HS14,HS:2015a,HS:2015b,HDW:2017} and the \raja~tool~\cite{HR:2013}.

We now overview alternatives to conducting amortised cost analysis through the means of a type-and-effect system.
The line of work~\cite{conf/sas/ZulegerGSV11,SZV:2014,SinnZV15,SinnZV17,conf/vmcai/FiedorHRSVZ18} has focused on identifying abstractions resp. abstract program models that can be used for the automated resource analysis of imperative programs.
The goal has been to identify program models that are sufficiently rich to support the inference of precise bounds and sufficiently abstract to allow for a scalable analysis,
employing the size-change abstraction~\cite{conf/sas/ZulegerGSV11}, (lossy) vector-addition systems~\cite{SZV:2014} and difference-constraint systems~\cite{SinnZV15,SinnZV17}.
This work has led to the development of the tool \loopus, which performs amortised analysis for a class of programs that cannot be handled by related tools from the literature.
Interestingly, \loopus\ infers worst-case costs from lexicographic ranking functions using arguments that implicitly achieve an amortised analysis (for details we refer the reader to~\cite{SinnZV17}).
Another line of work has targeted the resource bound analysis of imperative and object-oriented programs through the extraction of recurrence relations from the program under analysis, whose closed-form solutions then allows one to infer upper bounds on resource usage~\cite{AAGP08,AAGP:2011,conf/sas/Alonso-BlasG12,Flores-Montoya:2017}.
Amortised analysis with recurrence relations has been discussed for the tools \costa~\cite{conf/sas/Alonso-BlasG12} and \cofloco~\cite{Flores-Montoya:2017}.
Amortised analysis has also been employed in the resource analysis for rewriting~\cite{MS:2020} and non-strict function programs, in particular, if \emph{lazy evaluation} is conceived, cf.~\cite{JVFH:2017}.

Sublinear bounds are typically not in the focus of these tools, but can be inferred by some tools.
In the recurrence relations based approach to cost analysis~\cite{AAGP08,AAGP:2011} refinements of linear ranking functions are combined with criteria for divide-and-conquer patterns;
this allows their tool~\pubs\ to recognise logarithmic bounds for some problems, but examples such as \emph{mergesort} or \emph{splaying} are beyond the scope of this approach.
Logarithmic and exponential terms are integrated into the synthesis of ranking functions in~\cite{ChatterjeeFG17}, making use of an insightful adaption of Farkas' and Handelman's lemmas.
The approach is able to handle examples such as \emph{mergesort}, but again not suitable
to handle self-balancing data-structures.
A type based approach to cost analysis for an ML-like language is presented in~\cite{WWC17}, which uses the Master Theorem to handle divide-and-conquer-like recurrences.
Very recently, support for the Master Theorem was also integrated for the analysis of rewriting systems~\cite{WM:2020}, extending~\cite{AM:2016} on the modular resource analysis of rewriting to so-called logically constrained rewriting systems~\cite{FKN17}.
The resulting approach also supports the fully automated analysis of \emph{mergesort}.

We also mention the quest for abstract program models whose resource bound analysis problem is decidable, and for which the obtainable resource bounds can be precisely characterised.
We list here the size-change abstraction, whose worst-case complexity has been completely characterised as polynomial (with rational coefficients)~\cite{conf/mfcs/ColcombetDZ14,conf/csr/Zuleger15}, vector-addition systems~\cite{conf/lics/BrazdilCK0VZ18,conf/fossacs/Zuleger20}, for which polynomial complexity can be decided,
and LOOP programs~\cite{conf/fossacs/Ben-AmramH19}, for which multivariate polynomial bounds can be computed.
We are not aware of similar results for programs models that induce logarithmic bounds.

\section{Conclusion}
\label{Conclusion}

We have presented a novel amortised resource analysis based on the potential method. The
method is rendered in a type-and-effect system. Our type system has been designed with the goal of automation.
The novelty of our contribution is that this is the first approach towards an automation of
an \emph{logarithmic} amortised complexity analysis. In particular, we show how the precise logarithmic amortised cost
of \emph{splaying}---a central operation of Sleator and Tarjan's splay trees---can be checked semi-automatically in our system.
As our potential functions are logarithmic, we cannot directly encode the comparison between logarithmic expressions within the theory of linear arithmetic.
This however is vital for eg.\ expressing Schhoenmakers' and Nipkow's (manual) analysis~\cite{Schoenmakers93,Nipkow:2015}
in our type-and-effect system. In order to overcome this algorithmic challenge, we proposed several ideas for the \emph{linearisation} of the induced constraint satisfaction problem.
These efforts can be readily extended by
expanding upon the \emph{expert knowledge} currently employed, eg.\ via incorporation of the results of a static analysis performed in a pre-processing step.
In future work, we aim at extension of the developed prototypes to a fully automated analysis of logarithmic amortised complexity. Here it
may be profitable to expand the class of potential functions to take \emph{linear} potential functions into account. This does not invalidate our soundness theorem.

\paragraph*{In Memoriam.}

Martin Hofmann and the 4th author have discussed and developed a large part
of the theoretical body of this work together. Unfortunately, Martin's tragic hiking accident in January 2018
prevented the conclusion of this collaboration.
Due to Martin's great interest and contributions to this work, we felt it fitting to include him as first author.
We've tried our best to finalise the common conceptions and ideas. Still automation and continued research clarified
a number of issues and also brought a different focus on various matters of the material presented.
Martin Hofmann's work was revolutionary in a vast amount of fields and it
will continue to inspire future researchers---like he inspired us.


\end{document}